\documentclass{article}

\usepackage{graphicx} 

\usepackage{lmodern}
\usepackage[top=1in, bottom=1in, left=1in, right=1in]{geometry}

\usepackage{amsmath,amssymb,amsthm}
\usepackage{wrapfig}
\usepackage{tikz}
\usetikzlibrary{positioning}
\usepackage{braket}

\newtheorem{theorem}{Theorem}

\newtheorem{definition}{Definition}

\newtheorem{result}{Result}
\newtheorem{question}{Question}
\newtheorem{lemma}[theorem]{Lemma}
\newtheorem{fact}[theorem]{Fact}

\newtheorem*{lemma*}{Lemma}
\newtheorem*{corollary*}{Corollary}
\newtheorem*{theorem*}{Theorem}
\newcommand{\adeg}{\widetilde{\deg}}
\newcommand{\bdeg}{\operatorname{bdeg}}
\newcommand{\OR}{\mathrm{OR}}
\newcommand{\AND}{\mathrm{AND}}
\newcommand{\MAJ}{\mathrm{MAJ}}
\newcommand{\E}{\mathbb E}
\newcommand{\PrTH}{\operatorname{PrTH}}
\newcommand{\THR}{\operatorname{THR}}
\newcommand{\SF}{\mathcal{S}}

\usepackage{xcolor}
\usepackage[normalem]{ulem} 

\usepackage[
    backend=biber,
    style=numeric-comp,
    sorting=nty,
    url=false,
    maxbibnames=99
]{biblatex}
\usepackage[
    colorlinks=true,
    allcolors=blue,
    citecolor=blue
]{hyperref}
\usepackage[capitalise,nameinlink]{cleveref}

\title{A General Composition Theorem for Approximate Degree}

\author{%
\begin{tabular}{c@{\hspace{2.2em}}c}
Samruddhi Pednekar & Supartha Podder
\end{tabular}\\[0.9em]
{\small Stony Brook University, New York, USA}\\[0.4em]
{\small \texttt{\char`\{spednekar, supartha\char`\}@cs.stonybrook.edu}}
}

\date{}
\begin{document}

\maketitle

\begin{abstract}
A longstanding open question in Boolean function complexity asks whether
approximate degree composes multiplicatively under block composition.
Although a general multiplicative upper bound is known, matching lower
bounds have previously been established only for restricted classes of
functions. We resolve this question for all total Boolean functions by
proving the matching lower bound. Together with Sherstov's upper bound~\cite{robust_polynomials},
our result shows that, for every pair of total Boolean functions
$f:\{0,1\}^n\to\{0,1\}$ and $g:\{0,1\}^m\to\{0,1\}$,
\[
\widetilde{\deg}(f\circ g)
=
\Theta\!\left(
\widetilde{\deg}(f)\,\widetilde{\deg}(g)
\right),
\]
where $\widetilde{\deg}$ denotes constant-error approximate degree.

We also consider the analogous question for exact quantum query complexity and give a counterexample showing that multiplicative composition fails in this setting.
\end{abstract}

\tableofcontents

\newpage

\section{Introduction}

Approximate degree is a well-studied complexity measure of Boolean functions and a central tool for proving lower bounds on quantum query complexity~\cite{Paturi1992,NisanSzegedy1994,beals2001quantum,KulkarniPodder2016,aaronson2021degree,BunThaler2022,BunVoronova2023,ChakrabortyEtAl2024,PodderYaoYe2025,Sherstov2025,HuffstutlerEtAl2026, balodis2026optimal, BunJuvekarKing2026}. It is the minimum degree of a real polynomial that approximates
the function pointwise to within a fixed constant error,
typically $1/3$, on the Boolean cube. 
For total Boolean functions, tight composition theorems are known
for deterministic query complexity and polynomial degree
\cite{savicky2002determinism, montanaro2013composition, tal_function_composition}, as well as bounded-error quantum
query complexity \cite{reichardt2011reflections, kimmel2012quantum, lee2011quantum}. Recently,~\cite{meiburg2026perfectrandomizedcomposition} showed that constant-factor multiplicative composition fails for randomized query complexity even for total Boolean functions.
However, composition up to polylogarithmic factors remains open for randomized query
complexity and composition of approximate degree, even up to
polylogarithmic factors remained open prior to this work
\cite{ben2020tight, ben2022randomised, robust_polynomials, chakraborty_et_al:LIPIcs.APPROX/RANDOM.2023.63}.

In particular, the approximate-degree composition conjecture asks
whether, for all total Boolean functions $f$ and $g$,
\[
  \adeg(f \circ g) = \Theta(\adeg(f) \adeg(g)),
\]
where $f \circ g$ denotes block composition.

~\cite{robust_polynomials} resolved this question for one direction by proving $\adeg(f \circ g) = O(\adeg(f) \adeg(g))$ for all total Boolean functions $f$ and 
$g$.

Thus, the following question remained open prior to this work, 
\begin{question}\label{question_1}
\[\adeg(f \circ g) = \Omega(\adeg(f) \adeg(g))? \]
\end{question}
Several papers in the past two decades have shown that this bound holds true, for certain special classes of functions but prior to this work, the question still remained open when both inner and outer functions are arbitrary Boolean functions. \cite{sher_direct_product} proved the bound when the outer function has linear approximate degree. ~\cite{sherstov2013} and ~\cite{BUN20152} proved the lower bound for approximate degree of $\AND \circ \OR$ in terms of approximate degree of $\AND$ and approximate degree of $\OR$ using dual witness of $\OR$. Later, a few papers gave lower bounds for composed functions for special cases up to polylogarithmic bounds, thus focusing on the more relaxed version of $\adeg(f \circ g) = \tilde{\Omega}(\adeg(f)\adeg(g))$.

~\cite{classical_lower_bounds} proved the result for arbitrary inner function and symmetric outer function. They first gave a proof for $\OR$ as outer function and arbitrary $g$ as inner function by using a quantum algorithm for combinatorial group testing problem. They also mentioned the possibility of proving the same result using dual witnesses. Later, ~\cite{chakraborty_et_al:LIPIcs.APPROX/RANDOM.2023.63} gave a lower bound \[\adeg(f \circ g) \ge \tilde{\Omega}(\sqrt{bs(f)} \adeg(g)),\] where $bs(f)$ is the block sensitivity of $f$.  This result, which is in terms of a non-trivial complexity measure of $f$ and $\adeg(g)$ establishes the lower bound for any outer function $f$ which has $\sqrt{bs(f)}$ asymptotically equal to $\adeg(f)$. They also gave the result for an outer strongly $k-$junta symmetric function for $k = O(\sqrt{n})$. A more recent paper~\cite{ChakrabortyEtAl2024} established the result for recursive functions. They first used the following result
\[ \adeg(f \circ \MAJ_t \circ g) \ge \Omega(\adeg(f) \adeg(g) ), \; t = \Omega(\log{n}) \]

and then showed that the inserted majority function can be removed from the bound when $f$ or $g$ is a recursive function with a high depth and the base function of $g$ is of constant arity but is neither $\AND$ nor $\OR$. They also highlighted that proving $\adeg(f \circ \MAJ_t \circ g) \le \tilde{O}(\adeg(f \circ g))$ for $t = \Theta(\log{n})$ would solve Question~\ref{question_1} with polylogarithmic loss. They particularly mentioned in the open questions that giving a bound for arbitrary $f$ composed with $\OR$ will potentially give some insight on solving the conjecture.

In another recent work, \cite{cornelissen2026subcube} introduces \emph{subcube stifling}, a combinatorial measure $\mu(f)$, and proves that  $\widetilde{\deg}(f \circ g) = \Omega(\sqrt{\mu(f)}\,\widetilde{\deg}(g))$. They also show that random Boolean functions have $\mu(f) = \Theta(\log n)$ with high probability and construct examples from linear codes with $\mu(f) = \Theta(n)$.

Our work unconditionally proves the bound given in Question~\ref{question_1} for arbitrary outer and inner total Boolean functions. It first gives a lower bound on the approximate degree of an arbitrary Boolean function composed with $\OR$. It then extends it to arbitrary inner functions by giving a new probability distribution obtained from the dual witness of the inner function $g$ and using the properties of trigonometric polynomials and squarefree extraction elaborated later.

\begin{result}[Theorem~\ref{thm:composition-main}]
    Let $f: \{0,1\}^n \rightarrow \{0,1\}$ and $g: \{0,1\}^m \rightarrow \{0,1\}$ be arbitrary total Boolean functions. Then 
    \[ \adeg(f \circ g) = \Omega(\adeg(f) \adeg(g)). \]
\end{result}

Additionally, in this work we also show that exact quantum query complexity does not compose multiplicatively. It is known from prior work that bounded-error quantum query complexity composes.~\cite{negative_weights_adversary} showed that the general adversary bound, denoted by $\operatorname{Adv}^{\pm}(f)$ is a lower bound on bounded-error quantum query complexity and also that general adversary bound behaves well under function composition by giving the lower bound $\operatorname{Adv}^{\pm}(f \circ g) \ge \operatorname{Adv}^{\pm}(f) \operatorname{Adv}^{\pm}(g)$. Later,~\cite{span_programs_reichardt_09} proved the composition of general adversary bound by giving the upper bound $\operatorname{Adv}^{\pm}(f \circ g) \le  \operatorname{Adv}^{\pm}(f) \operatorname{Adv}^{\pm}(g)$ and also proved an upper bound on bounded-error quantum query complexity in terms of general adversary bound, up to logarithmic factors. This result was then strengthened by~\cite{reichardt2011reflections} giving a tight bound $Q_{\epsilon}(f) =\Theta( \operatorname{Adv}^{\pm}(f))$ for a fixed $\epsilon \in (0,1/2)$ and thus proving the composition of bounded-error quantum query complexity. Our result shows that a similar multiplicative lower bound does not hold for exact quantum query complexity of composition of total Boolean functions.

\begin{result}[Theorem~\ref{thm_Q_E_counterexample}]
    There exist total Boolean functions $f_k$ and integers $m_k$ such that
    \[ \frac{Q_E(f_k \circ \OR_{m_k})}{Q_E(f_k) Q_E(\OR_{m_k}) } \rightarrow 0. \]
    Thus, no multiplicative lower bound holds for the exact quantum query complexity.
\end{result}

\subsection{Overview of Proof Techniques}

\subsubsection{Proof Techniques for Composition of Approximate Degree}

The proof starts with an approximating polynomial $Q$ for $f\circ g$ and asks how to use it to approximate $f$. Each input bit of $f$ is the output of a block on which $g$ acts. If we sample each block so that its output under $g$ is a prescribed bit $y_i$, then averaging $Q$ over these blocks gives an approximation to $f(y)$. The challenge is to choose the sampling rules so that this averaging also captures the complexity of $g$. We will show that if the degree of $Q$ is too small, the resulting approximation to $f$ can be reduced to degree below $\adeg(f)$, giving a contradiction.

To make this idea precise, consider two distributions $\nu_0$ and $\nu_1$ on the inputs of $g$, chosen so that $X\sim\nu_b$ satisfies $g(X)=b$ with high probability. For an input $y=(y_1,\ldots,y_n)$ to $f$, sample the blocks $X^{(i)}$ independently from $\nu_{y_i}$. A union bound, together with the approximation guarantee for $Q$, shows that $\E Q(X^{(1)},\ldots,X^{(n)})$ approximates $f(y)$. When $\nu_b$ is supported on $g^{-1}(b)$, there is no error from the sampling. However, obtaining these approximations separately for each $y$ is not enough: we need them to arise from a single function whose degree or frequency we can control. This is why we introduce parameterized families of distributions. A similar averaging strategy has been used in~\cite{sherstov2013}.

The OR case in Section~\ref{section_inner_or} illustrates this idea directly. Sample each bit in block $X^{(i)}$ independently from a Bernoulli distribution with parameter $p_i$, and let $R(p_1,\ldots,p_n)$ be the expectation of $Q$. Setting $p_i=0$ forces the block output to be zero, while setting $p_i=a$ for a suitable $a=O(\log(2n)/m)$ makes the block output one with sufficiently high probability. Thus, setting $p=ay$ gives an approximation $R(ay)$ to $f(y)$. We may take $Q$ to be bounded in $[0,1]$ by error reduction, with only a constant factor increase in degree. Lemma~\ref{lem:bernoulli_avg} then shows that $R$ is also bounded in $[0,1]$ on $[0,1]^n$ and has degree at most $\deg Q$. The important point is that $R$ must approximate $f$ on a small scaled copy of the Boolean cube while remaining bounded on the entire cube. The dilation bound in Lemma~\ref{lem_or_dilation} shows that this requires
\[
\deg R=\Omega\!\left(\frac{\adeg(f)}{\sqrt{a}}\right),
\]
which gives the desired lower bound on $\deg Q$.

For promise threshold and arbitrary total inner functions, Sections~\ref{section_promise_threshold} and~\ref{section:inner_arbitrary} use a different averaging construction that produces a \emph{trigonometric polynomial} $R$. The approximation to $f$ is then recovered from its \emph{squarefree extraction} $\SF(R)$. The main task is to construct a family of probability measures $\{\rho_\theta:\theta\in\mathbb{R}\}$ that encodes the two output-conditioned distributions through its value and first derivative at zero:
\[
\rho_0=\frac{\nu_0+\nu_1}{2},
\qquad
\rho'_0=\frac{\kappa}{2}(\nu_1-\nu_0).
\]
The parameter $\kappa$ measures the size of this derivative. We want $\kappa$ to be large, while keeping the frequency of the averaged function controlled by the degree of $Q$. Achieving both properties is the main purpose of the constructions in Lemmas~\ref{lem:prob_dist_th} and~\ref{prob_inner_arbitrary}.

Let $D=\deg Q$, and sample the blocks independently from $\rho_{\theta_i}$. Lemmas~\ref{lem:threshold_avg_trig} and~\ref{squarefree_probability_outputclass_relation} show that
\[
R(\theta_1,\ldots,\theta_n)
=
\E_{X^{(i)}\sim\rho_{\theta_i}}
\left[Q(X^{(1)},\ldots,X^{(n)})\right]
\]
is bounded in $[0,1]$ on $\mathbb{R}^n$ and has total frequency at most $D$. To see how to recover $f$, observe that $\nu_{(1+z_i)/2}=\rho_0+(z_i/\kappa)\rho'_0$. Expanding the product of these measures across the blocks uses either the value or the first derivative in each block. These are precisely the terms selected by squarefree extraction from the Taylor expansion of $R$. Therefore,
\[
\SF(R)(z/\kappa)
=
\E_{X^{(i)}\sim\nu_{(1+z_i)/2}}
\left[Q(X^{(1)},\ldots,X^{(n)})\right]
\]
approximates $f((1+z)/2)$ for every $z\in\{-1,1\}^n$. Thus, the derivative scale $\kappa$ allows us to recover the outer function on a cube scaled by $1/\kappa$.

The remaining step is to show that this recovery is impossible when $D$ is too small. Write the Taylor expansion of $R$ at zero as $R=\sum_{r\ge0}H_r$, where $H_r$ is homogeneous of degree $r$. Lemmas~\ref{lem:squarefree_norm_bound} and~\ref{lem:homogeneous_taylor_coefficients} give
\[
\left\|
\SF(H_r)(\,\cdot\,/\kappa)
\right\|_{[-1,1]^n}
\le
\left(\frac{4eD}{\kappa r}\right)^r
\qquad (r\ge1).
\]
If $D$ is sufficiently small compared with $\kappa\adeg(f)$, the terms of degree at least $\adeg(f)$ contribute only a small error. We can therefore discard all these terms while retaining an approximation to $f((1+z)/2)$. After substituting $z=2y-1$, this gives an approximating polynomial for $f$ of degree below $\adeg(f)$, a contradiction. Consequently,
\[
D=\Omega\!\left(\kappa\adeg(f)\right).
\]
The construction of the probability family determines the value of $\kappa$, and hence the contribution of the inner function to the lower bound.

For our 2nd warm-up example, the promise threshold, the construction comes easily from the geometry of two adjacent Hamming layers. The function $\PrTH_m^{k+1}$ is zero on the layer of weight $k$ and one on the layer of weight $k+1$, so we take $\nu_0=\sigma_k$ and $\nu_1=\sigma_{k+1}$, where $\sigma_j=\operatorname{Uniform}\{x\in\{0,1\}^m:|x|=j\}$. Let $v_k$ and $v_{k+1}$ be the nonnegative unit vectors whose squared coordinates are these distributions, and set $\phi=(v_k+v_{k+1})/\sqrt{2}$. Applying $U_\theta^{\otimes m}$ to $\phi$ and squaring its coordinates gives the family $\rho_\theta$. The interaction between the adjacent layers gives
\[
\kappa=\sqrt{(k+1)(m-k)}
=
\Theta\!\left(
\widetilde{\bdeg}(\PrTH_m^{k+1})
\right).
\]
Here we use bounded approximate degree: $Q$ approximates the composition on its promise and remains bounded on the entire Boolean cube.

For an arbitrary total $g$, the construction is challenging. Its zero and one inputs need not have such a simple geometric relationship. A \emph{dual witness} turns out to be useful here: it supplies a replacement for the Hamming-layer structure. It gives nested subspaces corresponding to polynomial degrees and an associated operator $A$. Multiplication by a degree-$r$ polynomial connects only subspaces whose degree levels differ by at most $r$. Together with the choice of eigenvalues of $A$, this ensures that averaging such a polynomial produces frequency at most $r$. At the same time, the dual witness guarantees that the part of $A$ connecting the zero and one input subspaces has norm $\Theta(\adeg(g))$. Using singular vectors of this coupling, we choose a unit vector $v$ and define
\[
\rho_\theta(x)=\left|(e^{i\theta A}v)_x\right|^2.
\]
Since $e^{i\theta A}$ is unitary, these coordinates form a probability distribution. The coupling gives the required derivative identity with $\kappa=\Theta(\adeg(g))$, while the degree-subspace structure gives the frequency bound. Combining these two properties with the squarefree extraction and truncation argument yields the general composition lower bound.

\subsubsection{Proof Technique for Exact Quantum Query Complexity}

The proof first gives an exact quantum query algorithm for $f \circ \OR_m$ where $f:\{0,1\}^n \rightarrow \{0,1\}$ is an arbitrary total function. This exact quantum query algorithm first runs $n$ bounded error quantum query algorithms to compute the inner $\OR_m$ functions with some error and then corrects this error to get exact result. The total queries required for this is $O(m\deg(f)+ n^2 \sqrt{m})$. Then we can show that there exists a Boolean function $f_k$ and integer $m_k$, for which the desired lower bound does not hold. This is done by using the fact that exact quantum query complexity of $\OR_{m_k} : \{0,1\}^{m_k} \rightarrow \{0,1\}$ is $\Theta(m_k)$ and a function $f_k$ which is a $k$-fold composition of a specific $4$ bit function has $Q_E(f_k)=\Omega((5/2)^k)$ as shown in~\cite{ambainispolynomialdegree}. By choosing an appropriate $m_k$, we can get the ratio $\frac{Q_E(f_k \circ \OR_{m_k})}{Q_E(f_k)Q_E(\OR_{m_k})}$ that tends to be zero.

\subsection*{Acknowledgements}

We thank Srinivasan Arunachalam for bringing the question of whether exact quantum query complexity composes multiplicatively to our attention after the first version of this paper was posted on arXiv, and for helpful discussions.

\paragraph{AI Disclosure.}
We used ChatGPT (GPT-6 Astra) to assist with exploring research directions, identifying potential errors in arguments, simplifying
proof strategies, and improving the exposition. The counter-example to multiplicative composition of exact quantum query complexity was generated entirely by GPT-6.1 Sol Ultra. The authors take full responsibility for the mathematical content and final
presentation.

\section{Preliminaries}

\subsection{Notations}

We write $[n]=\{1,\ldots,n\}$ and $|x|=\sum_i x_i$ for the
Hamming weight of $x\in\{0,1\}^n$. 

$\mathbb{Q}$ is the set of rational numbers and $\mathbb{C}$ is the set of complex numbers.

For $x \in \mathbb{C}$, $\overline{x}$ is the complex conjugate of $x$.

For a finite set $\mathcal{X}$, we identify functions on $\mathcal{X}$
with vectors in $\mathbb{R}^{\mathcal{X}}$ and use the unnormalized inner product
$\langle u,v\rangle=\sum_{x\in\mathcal{X}}\overline{u_x}v_x$.
We write $\|v\|_1=\sum_x|v_x|$ and
$\|v\|_2=(\sum_x|v_x|^2)^{1/2}$.
For matrices, $\|A\|$ denotes the operator norm induced by $\|\cdot\|_2$;
for polynomials, $\|P\|_K=\sup_{x\in K}|P(x)|$.

We identify a function $\phi:\{0,1\}^m\to\mathbb{R}$ with the vector
$(\phi(x))_{x\in\{0,1\}^m}\in\mathbb{R}^{\{0,1\}^m}$.

We denote the diagonal multiplication operator associated with $h$
by $M_h$, so that $(M_hv)_x=h(x)v_x$.
We write $A^*$ for the conjugate transpose, $[A,B]=AB-BA$,
and $A\preceq B$ when $B-A$ is positive semidefinite.

We denote the standard orthonormal basis by
$\{e_x:x\in\{0,1\}^m\}$.
Under the identification
$\mathbb{R}^{\{0,1\}^m}\cong(\mathbb{R}^2)^{\otimes m}$,
we have $e_x=e_{x_1}\otimes\cdots\otimes e_{x_m}$,
where $e_0,e_1$ are the standard basis vectors of $\mathbb{R}^2$.

We will use the following rotation matrix in $\mathbb{R}^2$:
\begin{equation}\label{eq:rotation_matrix}
    U_{\theta} = \begin{pmatrix}
        \cos(\theta/2) & -\sin(\theta/2) \\
        \sin(\theta/2) & \cos(\theta/2)
    \end{pmatrix},
\end{equation}
and its tensor power that acts on $\mathbb{R}^{\{0,1\}^m}$ as, 
\[
    U_{\theta}^{\otimes m}
    =\underbrace{U_{\theta}\otimes\cdots\otimes U_{\theta}}_{m\text{ factors}}.
\]

We use Dirac notation for Section~\ref{section_exact_quantum}. For single qubits, we use standard notation:
\[ \ket{0} = \begin{pmatrix}
    1 \\ 0
\end{pmatrix}, \quad \ket{1} =  \begin{pmatrix}
    0 \\ 1
\end{pmatrix}, \quad \ket{+} =  \begin{pmatrix}
    1/\sqrt{2} \\ 1/\sqrt{2}
    \end{pmatrix}, \quad \ket{-} =  \begin{pmatrix}
    1/\sqrt{2} \\ -1/\sqrt{2} \end{pmatrix} \]

Additionally $\bra{\psi}$ denotes $\ket{\psi}^{*}$. When it is clear from context, $\ket{0}$ denotes $\ket{0}^{\otimes r}$ for appropriate number $r$ of qubits.

The Hadamard gate is represented by the matrix $H = \frac{1}{\sqrt{2}}\begin{pmatrix}
    1 & 1 \\
    1 & -1 \\
\end{pmatrix}$.

Please see~\cite{dewolf2026quantumcomputinglecturenotes},~\cite{hamoudi2025briefintroductionquantumquery} for more background on quantum computing preliminaries.
\subsection{Definitions}

We use the standard definition of approximate degree of Boolean functions which is stated below:
\begin{definition}[Approximate degree of Boolean functions]
    Let $f: \{0,1\}^n \rightarrow \{0,1\}$ be a Boolean function. For $0 \le \epsilon < 1/2$, a real polynomial $p: \mathbb{R}^n \rightarrow \mathbb{R}$ is said to approximate $f$ if 
    \[ |p(x)-f(x)| \le \epsilon, \qquad \text{for every} \; x \in \{0,1\}^n. \]
    The $\epsilon$-approximate degree of $f$, denoted by $\deg_{\epsilon}(f)$ is the minimum degree of all such $\epsilon$-approximating polynomials of $f$.

    In particular, $\adeg(f) = \deg_{1/3}(f)$. When $\epsilon=0$, it is the exact degree of $f$.
\end{definition}

\begin{definition}[Bounded approximate degree of Boolean functions]
    Let $f: D \subseteq \{0,1\}^n \rightarrow \{0,1\}$ be a Boolean function. For $0 \le \epsilon < 1/2$, a real polynomial $p: \mathbb{R}^n \rightarrow \mathbb{R}$ is said to approximate $f$ if 
    \[ |p(x)-f(x)| \le \epsilon, \qquad \text{for every} \; x \in D \quad \text{and} \quad 0\le p(x)\le 1, \; \text{for every} \; x \in \{0,1\}^n. \]
    The bounded $\epsilon$-approximate degree of $f$, denoted by $\bdeg_{\epsilon}(f)$ is the minimum degree of all such bounded $\epsilon-$approximating polynomials of $f$.

    In particular, $\widetilde{\bdeg}(f) = \bdeg_{1/3}(f)$.
\end{definition}

We use the standard block composition of Boolean functions, see~\cite{tal_function_composition}.
\begin{definition}[Composition of Boolean functions]
    Let $f: \{0,1\}^n \rightarrow \{0,1\}$ and $g: D \subseteq \{0,1\}^m \rightarrow \{0,1\}$. For $i \in [n]$, let $x^{(i)}  \in \{0,1\}^m$ and let $x = (x^{(1)},\dots,x^{(n)}) \in D^n $ where $x^{(1)},\dots,x^{(n)}$ are $n$ disjoint blocks of $m$ bits each. Then, the function composition, $f$ composed with $g$ is 
    \[ (f\circ g) (x) = f(g(x^{(1)}),\dots,g(x^{(n)})). \]
    $f$ is called the outer function and $g$ is called the inner function.
\end{definition}

\begin{definition}[Standard Oracle]
    An oracle is a unitary that gives black-box access to the input $x \in \{0,1\}^n$:
    \[ O_x \ket{i,b} = \ket{i,b\oplus x_i}, \]
    where $i \in [n]$ and $b \in \{0,1\}$.
\end{definition}

We use a modified version where for an additional index $i=0$, this oracle acts like an identity.

\begin{definition}[Quantum Algorithm]
    A quantum algorithm with $T$ queries is a sequence of input-independent unitary operators $U^0,\dots U^T$ interleaved with $T$ oracle unitaries $O_x$ and applied to a fixed state, let us say $\ket{0}$, producing a final state $U^T O_x\dots O_x U^0 \ket{0}$. This final state is measured to give the output of the algorithm. This quantum algorithm is bounded-error quantum query algorithm if it outputs the correct output with error probability at most $\epsilon \in [0,1/2)$.
\end{definition}

\begin{definition}[Quantum Query Complexity]
    The quantum query complexity of a function $f$, denoted by $Q_{\epsilon}(f)$ is the minimum number of queries needed by any quantum algorithm to compute $f(x)$ for every $x$ in the domain with an error probability at most $\epsilon \in [0,1/2)$. When $\epsilon = 1/3$, we use the notation $Q(f)$ and when $\epsilon =0$, it is called exact quantum query complexity and is denoted by $Q_{E}(f)$. 
\end{definition}

\begin{definition}[Threshold function]
    A threshold function, denoted as $\THR_m^t$, is defined below:
    \begin{equation*}
        \THR_m^t(x) = \begin{cases}
            1, \qquad |x|\ge t, \\
            0, \qquad \text{otherwise}.
        \end{cases}
    \end{equation*}
\end{definition}

\begin{definition}[Promise threshold function]
    A promise threshold function $\PrTH_m^t$ is defined below:
    \begin{equation*}
        \PrTH_m^t(x) = \begin{cases}
            1, \qquad &|x|= t, \\
            0, \qquad &|x|=t-1, \\
            \text{undefined}, \qquad &\text{otherwise}.
        \end{cases}
    \end{equation*}
\end{definition}

\begin{definition}[Symmetric function]
    A function $f:\{0,1\}^n \rightarrow \{0,1\}$ is symmetric if $f(x)=f(y)$ for all $x,y  \in \{0,1\}^n$ such that $|x|=|y|$.
\end{definition}

We use the standard definition of trigonometric polynomials. See~\cite{rudin} for more details.

\begin{definition}[Univariate trigonometric polynomial]
A univariate trigonometric polynomial is a function $R$ of the form 
\[ R(\theta) = a_0 + \sum_{n=1}^d a_n \cos{n \theta} + \sum_{n=1}^d b_n \sin{n \theta}, \; \theta \in \mathbb{R}, \; a_n \in \mathbb{C}, \; b_n \in \mathbb{C}, \; d \in \mathbb{Z}_{\ge0 }.\]

The degree of this trigonometric polynomial is $\max\{n: a_n \ne 0 \; \text{or} \; b_n \ne 0 \}$. $R$ is real-valued if $a_n, b_n \in \mathbb{R}$.

It is also written in the form
\[ R(\theta) = \sum_{n=-d}^d c_n e^{in \theta}.   \]
$R$ is real-valued if $c_{-n} = c_n^*$. 
\end{definition}

\begin{definition}[Squarefree extraction]
For a monomial $(x_1^{a_1}\cdots x_n^{a_n})$, where
$a_i$'s are nonnegative integers, define
\[
\SF(x_1^{a_1}\cdots x_n^{a_n})=
\begin{cases}
(x_1^{a_1}\cdots x_n^{a_n}), & \text{if $a_i\in\{0,1\}$ for every $i\in[n]$}\\
0 & \text{otherwise.}
\end{cases}
\]
Then for a polynomial $P=\sum_{j=1}^N c_j m_j$, where the $m_j$ are
distinct monomials and $c_j\in\mathbb{R}$, define
\[
\SF(P)=\sum_{j=1}^N c_j\,\SF(m_j).
\]

\end{definition}

This kind of operator is used in existing literature on arithmetic circuit complexity. For example, see the projection of polynomial onto multilinear monomials~\cite{mrinal2017}.

    Note that, squarefree extraction differs from Boolean multilinearization:
it deletes every monomial in which some variable has exponent
at least two, whereas Boolean multilinearization replaces every
positive exponent by one.

For a tuple $\alpha=(\alpha_1,\ldots,\alpha_n)$ of nonnegative
integers, write
\[
x^\alpha=\prod_{i=1}^n x_i^{\alpha_i},
\qquad
\alpha!=\prod_{i=1}^n(\alpha_i!),
\qquad
\partial^\alpha=\partial_1^{\alpha_1}\cdots\partial_n^{\alpha_n},
\qquad
|\alpha| = \sum_{i=1}^n \alpha_i.
\]
where $\partial_i$ denotes differentiation with respect to $x_i$.

We write $B_r(a)=\{x\in\mathbb{R}^n:\|x-a\|_2<r\}$ for the
open ball of radius $r$ centered at $a$.

A function $R:U\to\mathbb{R}$ on an open set
$U\subseteq\mathbb{R}^n$ is \emph{real analytic} if, for every
$a\in U$, it agrees with an absolutely convergent power series
centered at $a$ on some ball $B_r(a)\subseteq U$.

\begin{definition}[Taylor expansion at 0]
Let $R:U\to\mathbb{R}$ be real analytic, where
$U\subseteq\mathbb{R}^n$ is open and contains $0$.
The \emph{Taylor expansion of $R$ at $0$} is the series
\[
\sum_{\alpha\in\mathbb{Z}_{\ge 0}^n}
\frac{\partial^\alpha R(0)}{\alpha!}\,x^\alpha.
\]
For some $r>0$ with $B_r(0)\subseteq U$, this series converges
absolutely to $R(x)$ for every $x\in B_r(0)$.
\end{definition}

The squarefree extraction of this is

\[ \SF(R)(x) = \sum_{S \subseteq [n]} \partial_S R(0) \prod_{i \in S} x_i,  \]
where $S = \{i \in [n] : \alpha_i=1\}$ and $\partial_S = \prod_{i \in S} \partial_i$.

\begin{definition}[Total frequency of trigonometric polynomial]
    A multivariate trigonometric polynomial $R$ has total frequency at most $D$ if it can be written as
    \[ R(\theta) = \sum_{\|a\|_1 \le D} c_a e^{i \langle a, \theta \rangle }\]

    where $a \in \mathbb{Z}^n$ and $\|a\|_1 = \sum_i |a_i|$.
\end{definition}

We give the standard result from complex analysis below. For more background please refer to \cite{Conway1978}. 
\begin{fact}[Maximum modulus principle]
    Let $\mathcal{D} \subseteq \mathbb{C}$ be a bounded open connected set . Let $\overline{\mathcal{D}}$ be closure of $\mathcal{D}$. If $f$ is a holomorphic function on $\mathcal{D}$ and continuous on $\overline{\mathcal{D}}$, then $|f|$ attains its maximum on the boundary of $\mathcal{D}$.
\end{fact}

\begin{fact}\label{multilinear_coefficient_bound}
    For a Boolean function $f:\{0,1\}^n \rightarrow \{0,1\}$, if $H(y) = \sum_{S \subseteq [n]} c_S \prod_{i \in S} y_i$ is the unique multilinear polynomial representing $(-1)^{f(y)}$ for $y \in \{0,1\}^n$, then 
    \[ \sum_{S \subseteq [n]} |c_S| |S| \le n 3^n. \]
\end{fact}
\begin{proof}
    M\"obius inversion gives $c_S = \sum_{T \subseteq S} (-1)^{|S|-|T|} H(\mathbf{1}_T)$ and since $|H(\mathbf{1}_T)|=1$, we get $|c_S| \le 2^{|S|}$. Also,
    \[ \sum_{S \subseteq [n]} |S| 2^{|S|} = \sum_{k=0}^n \binom{n}{k} k 2^k  = \sum_{k=1}^n \binom{n}{k} k 2^k =\sum_{r=0}^{n-1} \binom{n}{r+1} (r+1) 2^{r+1}. \]

    Using $\binom{n}{r+1} (r+1) =n \binom{n-1}{r}$ and the binomial theorem, we get
    \[\sum_{S \subseteq [n]} |c_S| |S| \le 2n \sum_{r=0}^{n-1} \binom{n-1}{r} 2^r = 2n (1+2)^{n-1} = 2n3^{n-1} \le n3^n.\]
\end{proof}

\subsection{Basic Lemmas}

Proofs of Lemma~\ref{lem:squarefree_norm_bound}, Lemma~\ref{lem:homogeneous_taylor_coefficients} and Lemma~\ref{lem:squarefree_approximation} are deferred to \Cref{appendix-proof-priliminaries}.

\begin{lemma}[Error reduction~\cite{BunThaler2022}]\label{error_reduction}
    Let $F: \{0,1\}^n \rightarrow \{0,1\}$ be a Boolean function with $\adeg(F)=d$. Then, there exists a multilinear polynomial $q$ with degree at most $K_{\epsilon}d$ that approximates $F$ such that for $\epsilon \in (0,1/2)$, 
    \[ 0\le q(x) \le 1 \quad \text{for every x} \in \{0,1\}^n, \qquad |q(x) - F(x)| \le \epsilon \quad \text{for every x} \in \{0,1\}^n, \]
    where $K_{\epsilon}$ is a constant dependent on $\epsilon$.
\end{lemma}

\begin{lemma}[Norm of squarefree extraction]\label{lem:squarefree_norm_bound}
    Let $P$ be a real homogeneous polynomial of degree $r$ in $n$ variables. Then, 
    \[ \|\SF(P)\|_{[-1,1]^n} \le 4^r \|P\|_{[-1,1]^n} \]
\end{lemma}

\begin{lemma}[Homogeneous Taylor coefficients]\label{lem:homogeneous_taylor_coefficients} Suppose $R$ is a trigonometric polynomial of total frequency at most $D$ and $|R(\theta)| \le 1$ for all $\theta \in \mathbb{R}^n$. Taking Taylor expansion at zero as $R = \sum_{r\ge 0} H_r$, where $H_r$ is homogeneous of degree $r$. Then $|H_0| \le 1$ and for $r \ge 1$, 
\[ \|H_r\|_{[-1,1]^n} \le \left( \frac{e D}{r} \right)^r. \]
\end{lemma}

The concept of dual witness from dual linear programming has been previously used for showing lower bounds on approximate degree in~\cite{BUN20152} and~\cite{sherstov2013}.
\begin{fact}[Dual witness~\cite{BunThaler2022,BUN20152,sherstov2013}]\label{dual_witness_properties}
    Let $g: \{0,1\}^m \rightarrow \{0,1\}$ be a non-constant Boolean function. $\adeg(g) \ge d$ if and only if there exists a real function $\psi$ on $\{0,1\}^m$ such that
    \begin{align}
        & \|\psi\|_1=1 ; \label{dual_norm} \\
        & \langle \psi, p \rangle=0 \quad \text{for every polynomial $p$ of degree} <d; \label{high_pure_degree} \\
        & \langle \psi, g \rangle > 1/3. \label{correlation}
    \end{align}
\end{fact}

\begin{lemma}\label{lem:squarefree_approximation}
    Let $Q: \{0,1\}^{mn} \rightarrow \mathbb{R}$ be a polynomial. Let $\{ \rho_{\theta} : \theta \in \mathbb{R}\}$ be a differentiable family of probability measures on $\{0,1\}^m$ and let $\nu_0$ and $\nu_1$ be probability measures on $\{0,1\}^m$ such that, $\nu_{(1+z)/2} = \rho_0 + \frac{z}{\kappa} \rho_0'$ for $z \in \{-1,1\}$ and some scalar $\kappa \ne 0$. Let $X^{(1)}, \dots, X^{(n)}$ be $n$ disjoint blocks of $m$ bits each and $X^{(i)}$'s are sampled independently from $\rho_{\theta_i}$. Let $R$ be defined as 
    \[ R(\theta_1,\dots,\theta_n) = \E_{X^{(i)}\sim \rho_{\theta_i}}[Q(X^{(1)}, \dots, X^{(n)})]. \]
    If $R$ is analytic in a neighborhood of 0, then, 
    \[ \SF(R)\left(\frac{z}{\kappa}\right) =  \E_{X^{(i)}\sim \nu_{\frac{1+z_i}{2}}}[Q(X^{(1)}, \dots, X^{(n)})],
    \]
    for every $z \in \{-1,1\}^n$.
\end{lemma}

\section{Warm up: Lower Bounds for Special Classes of Inner Functions}

\subsection{Arbitrary Function composed with OR Function}
\label{section_inner_or}

We derive the following dilation bound from the classical higher-order Markov inequality; see~\cite{harris2002markov}. The proof controls homogeneous components through one-dimensional restrictions, a reduction also used in~\cite{robust_polynomials}.
Proofs of the following lemmas are deferred to \Cref{appendix-proof-f-or}.

\begin{lemma}[Dilation lemma for $\OR$]\label{lem_or_dilation}
    Let $f: \{0,1\}^n \rightarrow \{0,1\}$ be a non-constant Boolean function with $\adeg(f)=d_f \ge 1$. Let $R : \mathbb{R}^n \rightarrow \mathbb{R}$ be a polynomial with degree at most $d$. Additionally, $R(x)\in [0,1]$ for all $x \in [0,1]^n$. If for some $a \in [0,1]$,
    \[|R(ay) - f(y)| \le \frac{1}{10},\]
    for all $y \in \{0,1\}^n$, then 
    \[d \sqrt{a} > \frac{d_f}{4e}.\]
\end{lemma}

\begin{lemma}\label{lem:bernoulli_avg}
    Let $Q$ be a real-valued multilinear polynomial on the Boolean cube $\{0,1\}^{mn}$ with degree at most $d$. Let $X=(X^{(1)},\dots,X^{(n)})$ be $n$ disjoint blocks of $m$ bits each. Let $p=(p_1,\dots,p_n) \in [0,1]^n$.
    For all $mn$ bits, let $X^{(i)}_j$ be sampled independently from $Bernoulli(p_i)$. Then if $R$ is defined as
    \[ R(p_1,\dots,p_n) = \E[Q(X^{(1)},\dots,X^{(n)})], \]
    then $\deg(R) \le d$. Furthermore, if $0 \le Q \le 1$ for all $x \in \{0,1\}^{mn}$, then $0 \le R \le 1$ for $p \in [0,1]^n$. 
\end{lemma}

\begin{theorem}[Lower bound for $f$ composed with $\OR$]\label{theorem_inner_or}
    Let $f:\{0,1\}^n \rightarrow \{0,1\}$ be an arbitrary Boolean function. Then for every $m \ge 1$,
    \[ \adeg(f \circ \OR_m) \ge \Omega\left(\adeg(f) \sqrt{\frac{m}{\log{2n}}} \right). \]
    Equivalently, $\adeg(f \circ \OR_m) \ge \tilde{\Omega}(\adeg(f)\adeg(\OR_m))$.
\end{theorem}
\begin{proof}
    Let $F= f \circ \OR_m$. Let $\adeg(F)= d$ and $\adeg(f)=d_f$. For constant $f$, the result is immediate since $d_f=0$ and therefore, we focus on non-constant $f$. Let $q$ be a polynomial with degree $d$ that $1/3$- approximates $F$. Then after error amplification and from Lemma~\ref{error_reduction}, we get a multilinear polynomial $Q$ with degree at most $Kd$ such that 
    \[ 0 \le Q(x) \le 1, \quad |Q(x)-F(x)| \le \eta \quad \text{for every} \; x \in \{0,1\}^{mn}.  \] 
    where $\eta = 1/20$ and $K$ is an absolute constant.
    
    \[ R(p_1,\dots,p_n) = \E[Q(X^{(1)},\dots,X^{(n)})]. \]

    Let $p \in [0,1]^n$. For $i \in [n]$, let $X^{(i)}$ be an $m$-bit block such that $X^{(i)}_j$ is independently sampled from $\operatorname{Bernoulli}(p_i)$. Thus, $\Pr[X^{(i)}_j = 1] = p_i$. 
    Let $Y^{(i)} = \OR_m(X^{(i)})$ and $Y=(Y^{(1)},\dots,Y^{(n)})$.
    Since $|Q(X)-F(X)|\le \eta$, $\E|Q(X)-F(X)|\le \eta$. Also since,
    \[|\E[Q(X)-F(X)]| \le \E|Q(X)-F(X)| \le \eta,  \quad \E[Q(X)]=R(p), \quad \E[F(X)] = \E[f(Y)] \]
    we get 
    
    \begin{equation}
    |R(p)-\E[f(Y)]| \le \eta.        
    \end{equation}

    \[ \Pr[Y^{(i)}=0] = \Pr[X^{(i)}_1=0,\dots,X^{(i)}_m=0] = (1-p_i)^m. \]
    Let $\delta = \eta/n$ and $a= 1-{\delta}^{1/m}$.
    Let us now set $p=ay$ for a fixed arbitrary $y \in \{0,1\}^n$. The goal is to bound $\Pr[Y \ne y]$. 
    If $y_i=0$, then $\Pr[Y^{(i)}=y_i] = 1$. If $y_i=1$, then $\Pr[Y^{(i)}=0] = (1-a)^m$.
    \[ \Pr[Y \ne y] = \Pr\left[\bigcup_i (Y^{(i)} \ne y_i) \right] \le \sum_i \Pr[Y^{(i)} \ne y_i] \le n\delta = \eta.\]
    
    Since $f$ is a Boolean function $|f(Y)-f(y)| \le \mathbf{1}_{\{Y \ne y\}}$ and hence $\left|\E\left[f(Y) - f(y)\right]\right| \le \E\left[|f(Y) - f(y)|\right] \le \Pr[Y \ne y] \le \eta$.
  
    \[|R(p)-f(y)| \le |R(p)-\E[f(Y)]|+|\E[f(Y)]-f(y)| \le 2 \eta = 1/10.\]
    Substituting $p=ay$, we get $|R(ay)-f(y)| \le 1/10$. 
    Note that $y$ was arbitrary and therefore, $|R(ay)-f(y)| \le 1/10$ for all $y \in \{0,1\}^n$.
    Using Lemma~\ref{lem:bernoulli_avg}, we have $0 \le R \le 1$ for $p \in [0,1]^n$ and $\deg(R) \le K d$.
    
    Finally, by Lemma~\ref{lem_or_dilation}, we get 
    \[ Kd \sqrt{a} > \frac{d_f}{4e}.  \]

    Recall $a = 1- \delta^{1/m}$. Re-writing it as $a = 1-\exp{\log \delta^{1/m}} = 1-\exp{\frac{1}{m} \log \delta}$. Using $1-e^{-x} \le x$, we get \[ a \le \frac{1}{m}(-\log{\delta}) = \frac{1}{m} \log{\frac{n}{\eta}}. \]

    Hence, we get 
    \[ K d \sqrt{ \frac{1}{m} \log{\frac{n}{\eta}}} \ge Kd \sqrt{a} > \frac{d_f}{4e}. \]

    Finally, we get,
     \[  d > \frac{d_f \sqrt{m}}{4Ke \sqrt{\log(n/\eta)}}. \]
\end{proof}

\subsection{Arbitrary Function composed with Promise Threshold Function}
\label{section_promise_threshold}

 Because threshold functions are symmetric, one of the most natural probability distributions to use would be uniform distributions over Boolean strings with Hamming weights $k$ and $k+1$. More precisely, let $\sigma_k = \text{Uniform}\{x : |x|=k\}$. The support of this distribution $\sigma_k$ contains only the inputs with Hamming weight $k$. Thus, the output distributions here would be $\nu_0 = \sigma_k$ and $\nu_1 = \sigma_{k+1}$. 

The techniques in this section are analogous to the ones used in Section~\ref{section:inner_arbitrary} and thus give an easier understanding while transitioning from earlier proofs.

The omitted proofs of this section are deferred to  Appendix~\ref{appendix:exact_threshold}.

\begin{lemma}\label{lem:prob_dist_th}
    Let $m \ge 1$ and $0 \le k \le m-1$. Let $\kappa = \sqrt{(k+1)(m-k)}$. Let $\sigma_k = \operatorname{Uniform}\{x : |x|=k\}$, $\sigma_{k+1} = \operatorname{Uniform}\{x : |x|=k+1\}$ and let $\nu_0 = \sigma_k$ and $\nu_1 = \sigma_{k+1}$. 
    Define vectors $v_j$ for $j \in \{k,k+1\}$ and $\phi$ in $\mathbb{R}^{\{0,1\}^m}$ by
    \[ v_j = \binom{m}{j}^{-1/2} \sum_{|x|=j} e_x \;, \quad  \phi = \frac{v_k + v_{k+1}}{\sqrt{2}}. \] 
    Let $U_{\theta}$ be as defined in Equation~\ref{eq:rotation_matrix}. Define $\rho_{\theta}: \{0,1\}^m \rightarrow [0,1]$ by
    \[\rho_{\theta}(x) =  (\langle e_x, U_{\theta}^{\otimes m} \phi \rangle)^2.\]
    Then $\{\rho_{\theta}: \theta \in \mathbb{R}\}$ is a family of probability measures. Furthermore,
    \[ \rho_0 = \frac{\sigma_k + \sigma_{k+1}}{2}, \quad \rho_0' = \frac{\kappa}{2}(\sigma_{k+1}-\sigma_{k}). \]
    Consequently, for $z \in \{-1,1\}$, $ \rho_0 + \frac{z}{\kappa} \rho_0' = \sigma_{k+(1+z)/2} = \nu_{(1+z)/2}$. 
\end{lemma}

\begin{lemma}\label{lem:threshold_trig_expectation}
     Let $\phi \in \mathbb{R}^{\{0,1\}^m} $ be a unit vector and let $U_{\theta}$ be as defined above. Let $\rho_{\theta}(x) = (\langle e_x, U_{\theta}^{\otimes m} \phi \rangle)^2$ be a probability measure on $\{0,1\}^m$. Let $X = (X_1,\dots,X_m)$ such that $X \sim \rho_{\theta}$. Let $S \subseteq [m]$, $|S|=r$ and $X_S= \prod_{i\in S} X_i$ be a monomial. Then, $\E(X_S)$ is a trigonometric polynomial with frequency at most $r$. 
\end{lemma}

\begin{lemma}\label{lem:threshold_avg_trig} Let $Q$ be a multilinear polynomial on $n$ blocks of $m$ variables, with degree at most $D$. Also, let $0 \le Q \le 1$ on the entire Boolean cube. 
Let \[ R(\theta_1,\dots,\theta_n) = \E_{X^{(i)} \sim \rho_{\theta_i}} [Q(X^{(1)},\dots,X^{(n)})]
\]
where $X^{(i)}$ are blocks sampled independently from $\rho_{\theta_i}$ described in Lemma~\ref{lem:prob_dist_th}.

Then, $0 \le R \le 1$ on $\mathbb{R}^n$ and $R$ is a trigonometric polynomial of total frequency at most $D$.

Furthermore, for every $z \in \{-1,1\}^n$,
  \[ \SF(R)(z/\kappa) = \E_{X^{(i)} \sim \sigma_{k+(1+z_i)/2}}Q(X^{(1)},\dots,X^{(n)}). \]
\end{lemma}

\begin{theorem}[Lower bound for inner threshold]\label{thm:inner_th}
    For a non-constant arbitrary Boolean function $f$, and $1 \le t \le m$,
    \[\widetilde{\bdeg}(f \circ \PrTH_m^t) \ge \Omega\left(\adeg(f)\widetilde{\bdeg}(\PrTH^t_m)\right), \]
    where $\PrTH_m^t$ is defined as follows
    \[ \PrTH_m^t(x) = \begin{cases}
        0, &\text{if} \qquad |x|=t-1;\\
        1, &\text{if} \qquad |x|=t;\\
        \text{undefined}, \qquad &\text{otherwise}.
    \end{cases}\]
\end{theorem}
\begin{proof}
    Let $f \circ \PrTH_m^t = F$ and $\widetilde{\bdeg}(F) = D$ and $\adeg(f)=d_f$. Let $q$ be a $1/3$-approximating polynomial for $F$ with $\widetilde{\bdeg}(F)=\deg(q)$. Also, $0 \le q \le 1$ on the entire Boolean cube. Using the amplification polynomial from the proof of Lemma 1 in~\cite{buhrman_robust_polynomials}, we can get a polynomial that approximates $F$ within an error of $\eta=1/20$ and degree up to $KD$. Additional Boolean multilinearization on this polynomial does not increase the degree. Let the resultant multilinear polynomial be $Q$, it satisfies $0 \le Q \le 1$ on $\{0,1\}^{mn}$ and approximates $F$ within an error $\eta=1/20$ for $x$ in the promise of $F$. Let $k=t-1$. Take the family of probability measures $\{\rho_{\theta}: \theta \in \mathbb{R}\}$ constructed in Lemma~\ref{lem:prob_dist_th}. Then, for $z \in \{-1,1\}$ and $\kappa = \sqrt{(k+1)(m-k)}$,
    \[  \rho_0 + \frac{z}{\kappa} \rho_0' = \sigma_{k+(1+z)/2}. \]

    We can now construct $R$ such that 
    \[ R(\theta_1,\dots,\theta_n) = \E_{X^{(i)} \sim \rho_{\theta_i}} [Q(X^{(1)},\dots,X^{(n)})],
\]
$X=(X^{(1)},\dots,X^{(n)})$ where $X^{(i)}$ are blocks sampled independently from $\rho_{\theta_i}$. Then, using Lemma~\ref{lem:threshold_avg_trig} $R$ is a trigonometric polynomial with total frequency at most $KD$. We can write its Taylor expansion at 0, in terms of its homogeneous parts as 
\[ R = \sum_{r \ge 0} H_r. \]

We get its squarefree extraction as 
\begin{equation}\label{sq_homogeneous_R}
    \SF(R) = \sum_{r\ge 0} \SF(H_r)
\end{equation}

Using Lemma~\ref{lem:threshold_avg_trig}, we get,
\[ \SF(R)(z/\kappa) = \E_{X^{(i)} \sim \sigma_{k+(1+z_i)/2}}Q(X^{(1)},\dots,X^{(n)}). \]

 Fix an arbitrary $z \in \{-1,1\}^n$ and let $y \in \{0,1\}^n$ such that $y_i = (1+z_i)/2$.
Let $X=(X^{(1)},\dots,X^{(n)})$ where $X^{(i)}$ are blocks sampled independently from $\sigma_{k+(1+z_i)/2}$. Therefore, $\Pr[|X^{(i)}|=k+y_i]=1$. This implies that $\Pr[\PrTH_m^t(X^{(i)})=y_i]=1$.

 Let $Y=(Y^{(1)},\dots,Y^{(n)})$ where $Y^{(i)}=\PrTH_m^t(X^{(i)})$. Then $F(X)=f \circ \PrTH_m^t(X) = f(y)=f(\frac{1+z}{2})$.

Since $|Q(X) - F(X) | \le \eta$, 
\[|\E Q -f(y)|= |\E (Q(X) -F(X))| \le \E|Q-F|\le \eta. \]

 Finally, we get 
 \begin{equation}\label{bound_sq_f}
  \left|\SF(R)\left(\frac{z}{\kappa}\right) - f\left(\frac{1+z}{2}\right)\right| \le  \eta.   
 \end{equation}

Now, let $P$ be a polynomial as follows:
\[ P(z) = \sum_{r=0}^{d_f -1} \frac{1}{\kappa ^r} \SF(H_r) (z).\]

By writing $\SF(H_r)\left(\frac{z}{\kappa}\right) = \kappa^{-r} \SF(H_r)(z)$, we get
\begin{align*}
    |\SF(R)\left(\frac{z}{\kappa}\right) - P(z)| &= \left|\sum_{r\ge d_f} \kappa^{-r}\SF(H_r)\left(z\right)\right| \\
    &\le \sum_{r\ge d_f} {\kappa}^{-r} \|\SF(H_r)\|_{[-1,1]^n} \\
    &\le \sum_{r\ge d_f} 4^r {\kappa}^{-r} \|(H_r)\|_{[-1,1]^n} \\
    &\le \sum_{r \ge d_f} 4^r {\kappa}^{-r} \left( \frac{e KD}{r} \right)^r.
\end{align*}

We get the first equality from Equation~\ref{sq_homogeneous_R}, the first inequality is from triangle inequality, second inequality is by using Lemma~\ref{lem:squarefree_norm_bound} and the third inequality using Lemma~\ref{lem:homogeneous_taylor_coefficients}.

Let us assume that $KD \le \frac{d_f \kappa}{32e}$. Then,

\[|\SF(R)\left(\frac{z}{\kappa}\right) - P(z)| \le \sum_{r=d_f}^{\infty} 8^{-r} \le \frac{1}{7}.\]

Finally we get, 
\[\left|P(z) - f\left(\frac{\mathbf{1}+z}{2}\right)\right| < \frac{1}{3}.\]

Then the polynomial $p(y)=P(2y -\mathbf{1}) $ which has degree less than $d_f$ approximates $f(y)$ within an error of $1/3$ and this is a contradiction. Hence, $KD > \frac{d_f \kappa}{32e}$. Therefore,

\[ D \ge \Omega \left( d_f \sqrt{t (m-t+1)} \right).\]
Since $\widetilde{\bdeg}(\PrTH_m^t) = \Theta(\sqrt{t(m-t+1)})$ (see~\cite{nayak_wu}), the result follows.
\end{proof}

\section{Lower Bounds for Arbitrary Total Boolean Functions}
\label{section:inner_arbitrary}

The omitted proofs for this section are deferred to Appendix~\ref{appendix-proof-main-result}.

\begin{lemma}\label{lemma_def_deg_operator}
    Let $g:\{0,1\}^m \rightarrow \{0,1\}$ be a non-constant Boolean function with $\adeg(g)=d$. Then there exists a non-empty set $\mathcal{X} \subseteq \{0,1\}^m$, real orthogonal projectors $E_0,E_1,\dots,E_m$ on $\mathbb{R}^{\mathcal{X}}$ and a real symmetric matrix $A = \sum_{k=0}^m \tau_k E_k$ where $\tau_k= \min\{k,d\}$ such that they satisfy the following three properties:
    \begin{align}
        &\sum_{k=0}^m E_k = I, \qquad 0 \preceq A \preceq dI; \label{range_degree_operator} \\
        &E_kM_pE_{\ell} = 0 \quad \text{if} \; |k-\ell|>r \quad \text{for every polynomial $p$ of degree at most} \; r; \label{space_jump_lim} \\
        & \frac{d}{3} \le \lambda = \|P_0AP_1\| \le \frac{d}{2}, \label{scale_output_spaces} \; \text{where} \; P_b = M_{\mathbf{1}_{\{g=b\}}}.
    \end{align}
\end{lemma}

\begin{lemma}\label{prob_inner_arbitrary}
Let $\mathcal{X},A, \lambda,P_0,P_1$ be as defined in previous Lemma~\ref{lemma_def_deg_operator} and $\kappa = 2 \lambda$. There exist real unit vectors $u_0 \in P_0 \mathbb{R}^{\mathcal{X}}$ and $u_1\in P_1 \mathbb{R}^{\mathcal{X}}$ such that
     \[ P_0 A u_1 = \lambda u_0, \qquad  P_1 A u_0 = \lambda u_1. \]
    Define $v = \frac{u_0+i u_1}{\sqrt{2}}$ and for real $\theta$, define $\rho_{\theta}: \mathcal{X} \rightarrow [0,1]$ by
    \[\rho_{\theta}(x) = \left| (e^{i \theta A} v)_x \right|^2.\]
    Then $\{\rho_{\theta}: \theta \in \mathbb{R}\}$ is a family of probability measures.
    
    Also, there are probability measures $\nu_0$ and $\nu_1$ supported on $\mathcal{X} \cap g^{-1}(0)$ and $\mathcal{X} \cap g^{-1}(1)$ respectively, such that
    \[ \rho_0 = \frac{\nu_0 + \nu_1}{2}, \qquad \rho_0' =\frac{\kappa}{2} (\nu_1 - \nu_0). \]

    Consequently, for $z \in \{-1,1\}$
    \[ \nu_{\frac{1+z}{2}} = \rho_0 + \frac{z}{\kappa} \rho_0' . \]
\end{lemma}

\begin{lemma}\label{squarefree_probability_outputclass_relation}
    Let $Q$ be a real multilinear polynomial on $n$ disjoint blocks of $m$ bits each. Let degree of $Q$ be at most $D$ and $0 \le Q \le 1$ on the Boolean cube $\{0,1\}^{mn}$. Let $X^{(1)},\dots, X^{(n)}$ be $n$ independent disjoint blocks such that $X^{(i)}$ is sampled from $\rho_{\theta_i}$, where $\theta_i \in \mathbb{R}$ and $\rho_{\theta_i}$ are as described in the previous Lemma~\ref{prob_inner_arbitrary}. Let $R$ be defined as follows: 
    \[ R(\theta_1,\dots,\theta_n) = \E_{X^{(i)}\sim \rho_{\theta_i}}[{Q}(X^{(1)},\dots,X^{(n)})]. \]

    Then, $R$ is a real-valued trigonometric polynomial with total frequency at most $D$ and $0 \le R \le 1$ on $\mathbb{R}^n$. Furthermore, the squarefree extraction of $R$ satisfies:
    \begin{equation}
        \SF(R)(z/\kappa) = \E_{X^{(i)}\sim \nu_{(1+z_i)/2}} Q(X^{(1)},\dots,X^{(n)})  \qquad (z \in \{-1,1\}^n).
    \end{equation}
\end{lemma}

\begin{theorem}[Composition Theorem]\label{thm:composition-main}
    Let $f: \{0,1\}^n \rightarrow \{0,1\}$ and $g: \{0,1\}^m \rightarrow \{0,1\}$ be arbitrary total Boolean functions. Then 
    \[ \adeg(f \circ g) = \Omega(\adeg(f) \adeg(g)). \]
\end{theorem}
\begin{proof}
    Let $f \circ g = F$ and $\adeg(F) = D$ and $\adeg(f)=d_f$. If $f$ or $g$ are constant, then the result follows immediately. Therefore, let us see for non-constant $f$ and $g$. Let $q$ be a $1/3$-approximating polynomial for $F$ with $\adeg(F)=\deg(q)$. Let $\eta = 1/10$. Using Lemma~\ref{error_reduction}, there exists a multilinear polynomial $Q$ of degree at most $KD$, where $K$ is a constant, such that 
    \[0 \le Q \le 1, \qquad |Q(x)- f\circ g(x)| \le \eta, \qquad \text{for every} \; x \in \{0,1\}^{mn}.  \]

    Using Lemma~\ref{lemma_def_deg_operator} and Lemma~\ref{prob_inner_arbitrary}, construct a family of probability measures $\{\rho_{\theta}: \theta \in \mathbb{R}\}$ such that for $z \in \{-1,1\}$ and $\kappa =\Theta(\adeg(g))$,
    \[  \rho_0 + \frac{z}{\kappa} \rho_0' = \nu_{(1+z)/2}. \]
     By Lemma~\ref{prob_inner_arbitrary}, the probability measures $\nu_0$ and $\nu_1$ are supported on $\mathcal{X}_0 = \mathcal{X} \cap g^{-1}(0)$ and $\mathcal{X}_1 = \mathcal{X} \cap g^{-1}(1)$ respectively where $\mathcal{X}$ is as described in Lemma~\ref{lemma_def_deg_operator}.
    We can now construct $R$ such that 
    \[ R(\theta_1,\dots,\theta_n) = \E_{X^{(i)} \sim \rho_{\theta_i}} [Q(X^{(1)},\dots,X^{(n)})],
\]
$X=(X^{(1)},\dots,X^{(n)})$ where $X^{(i)}$ are blocks sampled independently from $\rho_{\theta_i}$. Then, using Lemma~\ref{squarefree_probability_outputclass_relation} $R$ is a trigonometric polynomial with total frequency at most $KD$. Also $0 \le R(\theta) \le 1$ for all $\theta \in \mathbb{R}^n$. We can write its Taylor expansion at 0, in terms of its homogeneous parts as 
\[ R = \sum_{r \ge 0} H_r. \]

Its squarefree extraction is 
\begin{equation}\label{square_homogeneous_2}
    \SF(R) = \sum_{r\ge 0} \SF(H_r)
\end{equation}

Using Lemma~\ref{squarefree_probability_outputclass_relation}, we get for every $z \in \{-1,1\}^n$,
\[ \SF(R)(z/\kappa) = \E_{X^{(i)} \sim \nu_{(1+z_i)/2}}Q(X^{(1)},\dots,X^{(n)}). \]

 Fix an arbitrary $z \in \{-1,1\}^n$ and let $y \in \{0,1\}^n$ such that $y_i = (1+z_i)/2$.
Let $X=(X^{(1)},\dots,X^{(n)})$ where $X^{(i)}$ are blocks sampled independently from $\nu_{(1+z_i)/2}$. Therefore, $\Pr[X^{(i)} \in g^{-1}((1+z_i)/2)]=1$ and thus, $\Pr[g(X^{(i)})=y_i]=1$.

 Let $Y=(Y^{(1)},\dots,Y^{(n)})$ where $Y^{(i)}=g(X^{(i)})$. Then $F(X)=f \circ g(X) = f(y)=f(\frac{\mathbf{1}+z}{2})$.

Since $|Q(X) - F(X) | \le \eta$, 
\[|\E Q -f(y)|= |\E (Q(X) -F(X))| \le \E|Q-F|\le \eta. \]

 Finally, we get 
 \begin{equation}\label{eq:squarefree_f_relation}
  \left|\SF(R)\left(\frac{z}{\kappa}\right) - f\left(\frac{1+z}{2}\right)\right| \le  \eta.   
 \end{equation}

Now, let $P$ be a polynomial as follows:
\[ P(z) = \sum_{r=0}^{d_f -1} \frac{1}{\kappa ^r} \SF(H_r) (z).\]

Therefore, $\deg(P)$ is at most $d_f-1$.

By writing $\SF(H_r)\left(\frac{z}{\kappa}\right) = \kappa^{-r} \SF(H_r)(z)$, we get
\begin{align*}
    |\SF(R)\left(\frac{z}{\kappa}\right) - P(z)| &= \left|\sum_{r\ge d_f} \kappa^{-r}\SF(H_r)\left(z\right)\right| \\
    &\le \sum_{r\ge d_f} {\kappa}^{-r} \|\SF(H_r)\|_{[-1,1]^n} \\
    &\le \sum_{r\ge d_f} 4^r {\kappa}^{-r} \|(H_r)\|_{[-1,1]^n} \\
    &\le \sum_{r \ge d_f} 4^r {\kappa}^{-r} \left( \frac{e KD}{r} \right)^r.
\end{align*}

We get the first equality from Equation~\ref{square_homogeneous_2}, and the subsequent inequalities from triangle inequality, Lemma~\ref{lem:squarefree_norm_bound} and Lemma~\ref{lem:homogeneous_taylor_coefficients}.

Let us assume that $KD \le \frac{d_f \kappa}{32e}$. Then,

\begin{equation}\label{eq:squarefree_auxiliary_relation}
    |\SF(R)\left(\frac{z}{\kappa}\right) - P(z)| \le \sum_{r=d_f}^{\infty} 8^{-r} \le \frac{1}{7}.
\end{equation}

Finally from Equation~\ref{eq:squarefree_f_relation} and Equation~\ref{eq:squarefree_auxiliary_relation}, 
\[\left|P(z) - f\left(\frac{\mathbf{1}+z}{2}\right)\right| < \frac{1}{3}.\]

Then the polynomial $p(y)=P(2y -\mathbf{1}) $ which has degree less than $d_f$ approximates $f(y)$ within an error of $1/3$ and this is a contradiction. Hence, $KD > \frac{d_f \kappa}{32e}$. Therefore,

\[ D \ge \Omega \left( d_f \kappa \right).\]
Since $\adeg(g) = \Theta(\kappa)$, the result follows.
\end{proof}

\section{Counter-Example for Exact-Quantum Query Composition}
\label{section_exact_quantum}
We show that exact quantum query complexity does not compose multiplicatively. It suffices to find a function where such multiplicative lower bound does not hold and this is shown in Theorem~\ref{thm_Q_E_counterexample}.

The first step is to construct an exact quantum query algorithm for $f \circ \OR_m$ by using bounded-error quantum query algorithms for inner $\OR_m$ blocks and then correcting the resulting error.

\begin{theorem}\label{thm_upper_bound_Q_f_or}
    For every Boolean function $f: \{0,1\}^n \rightarrow \{0,1\}$ of degree $d$, 
    \[ Q_E(f \circ \OR_m) = O(md + n^2\sqrt{m}). \]
\end{theorem}

\begin{proof}
    Let $x = (x^{(1)},\dots,x^{(n)})$ where $x^{(i)} \in \{0,1\}^m$ for $i \in [n]$ be an input to $f \circ \OR_m$. Also, define $y_i = \OR_m(x^{(i)})$ for $i \in [n]$. For $z \in \{0,1\}^n$, let $H(z) = (-1)^{f(z)}$ and thus we can write it as a unique multilinear polynomial as 
    \[H(z) = \sum_{S \subseteq [n]} c_S \prod_{i \in S} z_i. \] 
    Since $H(z)$ can also be written as $1-2f(z)$, $\deg(H)\le d$.

    Let $A_i$ be a bounded error quantum algorithm for $\OR_m$ with error at most $\epsilon = 2^{-4n}$ and $t$ queries. Thus, $t = O(\sqrt{m} \log{(1/\epsilon})) = O(n\sqrt{m})$. Let $a_i$ be the probability that $A_i$ outputs $1$. Run these algorithms $(A_1,\dots,A_n)$ on the respective $\OR_m$ input blocks and these algorithms have separate output registers. Also, let $A =A_1 \otimes \dots \otimes A_n $.
    Now, let $b = (b_1,\dots,b_n) \in \{0,1\}^n$ be the computational basis values of output registers of $A = A_1 \otimes \dots \otimes A_n$. Thus,
    \[ |\Pr[b_i=1]-y_i|=|a_i - y_i| \le \epsilon, \; \text{for all inputs} \; x^{(i)} \in \{0,1\}^m. \]

    Let the probability of observing $b$ is $p(b)$.

    \[ p(b) = \prod_{i } a_i^{b_i} (1-a_i)^{1-b_i}. \]
     Let $D_f = \sum_{ b\in \{0,1\}^n} (-1)^{f(b)} \ket{b} \bra{b} $ be a phase operator and applying this to the output registers of $A$, we get,
     \[ D_f A \ket{0} = D_f \sum_{b} \alpha_b \ket{b} \ket{\psi_b}, \]
        where $|\alpha_b|^2 = p(b)$.
    Therefore, $D_f A \ket{0} = \sum_{b} (-1)^{f(b)}\alpha_b \ket{b} \ket{\psi_b}$
    Now, applying the reverse algorithm $A^{*}$ and then calculating the amplitude of returning to $\ket{0}$,  
    \[ \bra{0} A^{*} D_f A \ket{0} = \sum_b |\alpha_b|^2 (-1)^{f(b)} = \sum_b p(b)(-1)^{f(b)}.  \]
     
     Also,
     \[ \E[(-1)^{f(b)}]  = \sum_{b} p(b)(-1)^{f(b)}.\]

     Now, let $B = (B_{1},\dots, B_{n})$ where $B_{i}$ are sampled independently from $\operatorname{Bernoulli}(a_i)$. Then by linearity of expectation, independence of $B_i$'s and $\E(B_i) = \Pr[B_i=1]=a_i$, we get,
     
     \[ \E[H(B)] = \E\left[\sum_{S \subseteq [n]} c_S \prod_{i \in S} B_i\right] =\sum_{S \subseteq [n]} c_S \E\left[\prod_{i \in S} B_i \right]=  \sum_{S \subseteq [n]} c_S \prod_{i \in S} \E[B_i] = \sum_{S \subseteq [n]} c_S \prod_{i \in S} a_i = H(a_1,\dots,a_n). \]

     Thus, finally we get,
     \[ \bra{0} A^{*} D_f A \ket{0} = H(a_1,\dots,a_n).  \]
    Let $B_x$ be the unitary defined by $A^{*} D_f A$ and thus the cost of $B_x$ is $2nt$ queries.
     We will now see the difference $H(y) - H(a) = \sum_{S \subseteq [n]} c_S \left(\prod_{i \in S} y_i -  \prod_{i \in S} a_i\right)$. For a non-empty set $S$ with $|S|=s$, index the elements with $i_1,\dots,i_s$. Then the difference in the monomials $\prod_{i \in S} y_i - \prod_{i \in S} a_i$ can be written as $\prod_{j \in [s]} y_{i_j} - \prod_{j \in [s]} a_{i_j}$ respectively. Further we can rewrite this as telescoping series and get 
     \[ \sum_{j=1}^s (y_{i_j} - a_{i_j}) \prod_{k < j} y_{i_k} \prod_{k > j} a_{i_k}. \]

     Then, for every non-empty set $S$ with $|S|=s$ and $j \in [s]$, define $\tau_{S,j}(x) = \frac{(y_{i_j} - a_{i_j})}{\epsilon}  \prod_{k < j} y_{i_k} \prod_{k > j} a_{i_k} $.

     Hence,
     \[H(y) - H(a) = \epsilon \sum_{S \subseteq [n]} c_S \sum_{j=1}^{|S|} \tau_{S,j}(x). \]

     Note that since $\tau_{S,j}(x) \in [-1,1]$, we can find a unitary $U_{S,j,x} $ such that $\bra{0}U_{S,j,x}\ket{0} = \tau_{S,j}(x)$.

     Let $x_S$ be a concatenated string $(x^{(i)})_{i \in S}$ where each $x^{(i)} \in \{0,1\}^m$.
     Let $Q_{S,x}\ket{0} = \ket{x_S}$. This unitary loads the bits of $x$ belonging to blocks indexed by $S$ into a register. This unitary $Q_{S,x}$ requires $m|S|$ queries. Additionally, let 
     \[ R_{S,j} (z)= \begin{pmatrix}
         \tau_{S,j}(z) & \sqrt{1-(\tau_{S,j}(z))^2} \\
         \sqrt{1-(\tau_{S,j}(z))^2} & -\tau_{S,j}(z)
     \end{pmatrix}, \quad R_{S,j} = \sum_{z \in \{0,1\}^{m|S|}} \ket{z} \bra{z} \otimes R_{S,j}(z). \]
     Let $U_{S,j,x} = Q_{S,x}^{*}R_{S,j}Q_{S,x}$. Then $\bra{0}U_{S,j,x}\ket{0}=\tau_{S,j}(x)$.
     Since $R_{S,j}$ is input independent, $U_{S,j,x}$ costs at most $2m|S|$ queries. Since $c_S=0$ for $|S|>d$, we get the query cost for $U_{S,j,x}$ is at most $2md$.

     Next step is to combine $B_x$ and $U_{S,j,x}$ coherently so that we get a matrix element proportional to $H(y)$.
     Now, set $\eta = \epsilon \sum_{S \subseteq [n]} |c_S||S|$. Since $\sum_{S \subseteq [n]} |c_S||S| \le n 3^n$ from Fact~\ref{multilinear_coefficient_bound}, we get $\eta \le 2^{-4n}n 3^n < 1$.
    We now need a selector-register superposition with desired weights that add up to some constant and a controlled unitary to control the unitary acting on the input. Since the selector register is independent of $x$, let $\ket{B}$ be the basis element corresponding to $B_x$ and $\ket{S,j}$ be the basis element corresponding to $U_{S,j,x}$. If we assign weight $1$ to $B$ and $\epsilon |c_S|$ to $\ket{S,j}$, then to balance the total weight, we need one more unitary $F$ such that $\bra{0} F \ket{0} = 0$ with weight $1- \eta$. Note that $F$ is input-independent. Define $G$ with normalized weights so that 
    \[ G\ket{0}  = \frac{1}{\sqrt{2}}\ket{B}+ \sum_{{S,j},c_S \ne 0} \sqrt{\frac{\epsilon |c_S|}{2}} \ket{S,j}+ \sqrt{\frac{1-\eta}{2}} \ket{F}. \]
    Let the controlled unitary $C_x$ be defined as follows:
    \[  C_x =  \ket{B} \bra{B} \otimes B_x + \sum_{S \subseteq [n], c_S \ne 0} \sum_{j=1}^{|S|} \ket{S,j} \bra{S,j} \otimes \operatorname{sgn}(c_S) U_{S,j,x} + \ket{F} \bra{F} \otimes F . \]
    By adding unused qubits, we can consider $B_x, U_{S,j,x},F$ as acting on a common workspace. The query-dependent part of $C_x$ is implemented in two sequential stages, both controlled by the selector register. For the state $\ket{B}$, $C_x$ implements $B_x$ which requires $2nt$ queries. If the state is $\ket{S,j}$, then $U_{S,j,x}$ is implemented. A single $O_x$ application can perform the corresponding query on all the components of the superposition and thus the second stage makes at most $2md$ oracle calls in total by implementing all $U_{S,j,x}$ coherently. Since $U_{S,j,x}$ has different query cost for different $S$, for $U_{S,j,x}$ making fewer than $2md$ queries, a dummy address can be used to indicate no-query for the remaining queries. For example, for a dummy address $0$, $O_x\ket{0,b} = \ket{0,b}$. 
    Finally, define $W_x = (G \otimes I)^{*} C_x (G \otimes I)$.

     We will finally get 
     \[ \bra{0} W_x \ket{0} = \frac{1}{2}H(a) + \frac{1}{2} \epsilon \sum_{S,j} c_S \tau_{S,j}(x) = \frac{1}{2} H(y). \]
     We still need to correct this to $H(y)$. Also, note that the cost of queries by $W_x$ is at most $2nt+2md$.
     Let $R = I - 2 \ket{0}\bra{0}$ and define $V_x = -W_x R W_x^{*}R W_x$. Then,
     $\bra{0} V_x \ket{0} = -\bra{0}W_x R W_x^{*}R W_x \ket{0}$. 

     We have $R W_x \ket{0} = (I - 2 \ket{0}\bra{0}) W_x \ket{0} = W_x \ket{0} - H(y) \ket{0} $. Also, $ \bra{0} W_x R  = \bra{0} W_x (I - 2 \ket{0}\bra{0})  = \bra{0} W_x - H(y) \bra{0} $. Therefore, $\bra{0} W_x R W_x^{*} = \bra{0} - H(y) \bra{0} W_x^{*}$. Thus,

     \begin{align*}
       \bra{0} W_x R W_x^{*} R W_x \ket{0} &=  \bra{0}W_x \ket{0} - H(y) \bra{0}
       W_x^{*} W_x \ket{0} -H(y) \bra{0} \ket{0} + H(y)^2 \bra{0} W_x^{*} \ket{0} \\
       &= \frac{H(y)}{2} - 2 H(y) + \frac{H(y)^3}{2} \\
       &= \frac{H(y)^3 -3 H(y)}{2}  \\
       &= -H(y).
     \end{align*}
     The last equality comes from the fact that $H(y) \in \{-1,1\}$. Hence, $\bra{0} V_x \ket{0} = H(y)$.
   Since $V_x$ is unitary, $V_x \ket{0}$ is a unit vector. Since $|\bra{0} V_x \ket{0}|=  |H(y)|= 1$, we get 
    \[V_x \ket{0} = H(y) \ket{0} = (-1)^{f(y)} \ket{0}.\]

    Now let us take the controlled version of $V_x$, that is $CV_x = \ket{0}\bra{0}\otimes I + \ket{1} \bra{1} \otimes V_x$. Applying $CV_x$ to $\ket{+}\ket{0}$, and then applying Hadamard on it gives $\ket{f(y)}\ket{0}$. Measuring the first register would give you the required output. This controlled version of $V_x$ costs the same number of queries as $V_x$.
    The reflection unitary $R$ does not contribute to the number of queries because they are input independent. But since $V_x$ involves invoking $W_x$ three times, two times $W_x$ and one time $W_x^{*}$, the cost of $V_x$ is 3 times cost of $W_x$. Since the cost of $W_x$ is at most $2nt+ 2md$, we get the query complexity this exact algorithm is $O(md+n^2\sqrt{m})$.
\end{proof}

\begin{theorem}\label{thm_Q_E_counterexample}
    There exist total Boolean functions $f_k$ and integers $m_k$ such that
    \[ \frac{Q_E(f_k \circ \OR_{m_k})}{Q_E(f_k) Q_E(\OR_{m_k}) } \rightarrow 0. \]
    Thus, no multiplicative lower bound holds for the exact quantum query complexity.
\end{theorem}
\begin{proof}
    Let $h(x_1,x_2,x_3,x_4) = x_1+x_2+x_3x_4-x_1x_2-x_2x_3-x_1x_4$ and let $f_k(x) = \underbrace   {h\circ \dots \circ h}_{k\text{-fold composition}}(x)$ where $x \in \{0,1\}^{n_k}$ with $n_k = 4^k$. Note that $\deg(f_k) = 2^k$. Using a result from~\cite{ambainispolynomialdegree}, we get
    \[ Q(f_k) = \Omega((5/2)^k). \]

    Since $Q_E(g) \ge Q(g)$ for all Boolean functions $g$, we get $Q_E(f_k) = \Omega((5/2)^k) $. Taking $m_k = n_k^3 = 64^k$, and applying Theorem~\ref{thm_upper_bound_Q_f_or}, we get,
    \[ Q_E(f_k \circ \OR_{m_k}) = O(m_k 2^k + (n_k)^2\sqrt{m_k}) = O(128^k). \]

    We know from~\cite{beals2001quantum} $Q_E(\OR_{m_k}) = m_k$. Therefore, 
    \[Q_E(f_k)Q_E(\OR_{m_k}) = \Omega((5/2)^k 64^k) = \Omega(160^k). \]

    Therefore, $\frac{Q_E(f_k \circ \OR_{m_k})}{Q_E(f_k) Q_E(\OR_{m_k})} = O((4/5)^k)$ and thus for $k \rightarrow \infty$, this quantity tends to 0.
\end{proof}

\printbibliography
\appendix

\section{Proofs of Lemmas from the Preliminaries}\label[appendix]{appendix-proof-priliminaries}

\begin{lemma*}[Restatement of Lemma~\ref{lem:squarefree_norm_bound}]
    Let $P$ be a real homogeneous polynomial of degree $r$ in $n$ variables. Then, 
    \[ \|\SF(P)\|_{[-1,1]^n} \le 4^r \|P\|_{[-1,1]^n} \]
\end{lemma*}
\begin{proof}
    If $r\le 1$, then $P$ is already in squarefree form and hence, the result is immediate. Let us now consider the case when $r \ge 2$. 
    We will now consider $P$ with same coefficient but in $\mathbb{C}^n$. For a fixed $z \in \mathbb{C}^n$ such that $|z_i| \le 2$ for all $i \in [n]$, we will define a real $n-$dimensional vector $x(t)$ such that for $t \in \mathbb{R}$, 
    \[ x_z(t) = \frac{1}{2} \text{Re}(e^{it}z). \]

    Note that $x_z(t) \in [-1,1]^n$. Therefore, $\max_{t \in [0,2 \pi]} |P(x_z(t))| \le \|P\|_{[-1,1]^n}$. 

    Now, we can write $x_z(t) = \frac{e^{it}z + e^{-it}\overline{z}}{4}$ and hence, $(x_z(t))_j = \frac{e^{it}z_j + e^{-it}\overline{z_j}}{4}$. Since $P$ is a homogeneous polynomial of degree $r$, we can write it as 
    \[ P(x_z(t)) = \sum_{|\alpha|=r} c_{\alpha}(x_z(t))^{\alpha} = \sum_{|\alpha|=r} c_{\alpha} \prod_j (x_{z}(t))_j^{\alpha_j}.  \]

    Since $(x_z(t))^{\alpha} = \prod_j (x_{z}(t))_j^{\alpha_j} = \prod_j \left(\frac{e^{it}z_j + e^{-it}\overline{z_j}}{4}\right)^{\alpha_j}$ and $\sum_j \alpha_j = r$, if we have a term with frequency $r$, it comes from $\frac{e^{it}z_j}{4}$ since $\frac{e^{-it}\overline{z_j}}{4}$ will only reduce the frequency. Therefore, the term with frequency $r$ is $\frac{e^{irt}z^{\alpha}}{4^{r}}$. Therefore, the coefficient of frequency $r$-term of $(x_z(t))^{\alpha}$ is $4^{-r}z^{\alpha}$. Substituting this back in $P(x_z(t))$, we get that the coefficient of frequency $r$-term is $4^{-r}P(z)$. Finally, we get 
    \begin{equation}\label{intermediate_bound}
        4^{-r}|P(z)| \le  \|P\|_{[-1,1]^n}.
    \end{equation} 

    Let $\Theta_i$ be a random variable sampled independently with density $p(\theta) = \frac{1+\cos \theta}{2 \pi}$ for $\theta \in [0,2 \pi]$. Let $Z_i = 2e^{i \Theta_i}$
    \[ \E[Z_i^0] = 1, \qquad \E[Z_i^1] = 1, \qquad \E[Z_i^a] = 0 \quad \text{for every integer} \; a \ge 2. \]

    We will now see $\E[P(Z_1y_1,\dots,Z_ny_n)] = \SF(P)(y_1,\dots,y_n)$.

    \[ \E[P(Z_1y_1,\dots,Z_ny_n)] = \E\left[\sum_{|\alpha|=r} c_{\alpha} \prod_j (Z_jy_j)^{\alpha_j} \right] = \sum_{|\alpha|=r} c_{\alpha} \E\left[ \prod_j (Z_jy_j)^{\alpha_j}\right]. \]

    Since $Z_j$ are sampled independently, 
    \[ \E[P(Z_1y_1,\dots,Z_ny_n)] =  \sum_{|\alpha|=r} c_{\alpha} \prod_j \E[  Z_j^{\alpha_j}]{y_j}^{\alpha_j} = \sum_{\alpha \in \{0,1\}^n} c_{\alpha} \prod_j {y_j}^{\alpha_j} . \]
    But $\sum_{\alpha \in \{0,1\}^n} c_{\alpha} \prod_j {y_j}^{\alpha_j} = \SF(P)(y)$. Therefore, we get $\E[P(Z_1y_1,\dots,Z_ny_n)] = \SF(P)(y_1,\dots,y_n)$.

    For $y \in [-1,1]^n$, $|\SF(P)(y_1,\dots,y_n)| = |\E[P(Z_1y_1,\dots,Z_ny_n)]| \le \E[|P(Z_1y_1,\dots,Z_ny_n)|] $.
    
    Since $|Z_i| = 2$ and $|y_i|\le 1$, $|Z_i y_i| \le 2$ and hence from Equation~\ref{intermediate_bound},
    \[\|\SF(P)\|_{[-1,1]^n} \le 4^{r} \|P\|_{[-1,1]^n}.\]
\end{proof}

\begin{lemma*}[Restatement of Lemma~\ref{lem:homogeneous_taylor_coefficients}] Suppose $R$ is a trigonometric polynomial of total frequency at most $D$ and $|R(\theta)| \le 1$ for all $\theta \in \mathbb{R}^n$. Taking Taylor expansion at zero as $R = \sum_{r\ge 0} H_r$, where $H_r$ is homogeneous of degree $r$. Then $|H_0| \le 1$ and for $r \ge 1$, 
\[ \|H_r\|_{[-1,1]^n} \le \left( \frac{e D}{r} \right)^r. \]
\end{lemma*}

\begin{proof}
Let us first look at the univariate version. Define $v: \mathbb{C} \rightarrow \mathbb{C}$ as
\[ v(z)= \sum_{j=-N}^N a_je^{ijz}, \qquad |v(z)| \le 1 \; \text{for} \; z \in \mathbb{R}.  \]

Also define 
\[ L(w)= \sum_{j=-N}^N a_j w^j, \qquad P(w) = w^N L(w). \]

Note that $L(e^{iz}) = v(z)$.

On the unit circle, $w = e^{iz}$ for $z \in \mathbb{R}$. Therefore, on the unit circle, $L(w) = \sum_{j=-N}^N a_j e^{ijz} =v(z)$ and hence, $|L(w)| \le 1$. Consequently, $|P(w)| \le 1$ on the unit circle. Using maximum modulus principle, $|P(w)| \le 1$ for $|w| \le 1$. Hence, $|L(w)| \le |w|^{-N}$ for $|w| \le 1$ and $w \ne 0$.

Now, look at the polynomial $\tilde{P}(w) = w^{2N} P(1/w) $. Thus, $\tilde{P}(w) = w^{N} L(1/w) $. $|L(1/w)| \le 1$ on the unit circle and hence $|\tilde{P}(w)| \le 1$ on the unit circle. Using maximum modulus principle, $|\tilde{P}(w)| \le 1$ on $|w| \le 1$. Hence, $|L(w)| \le |w|^{N}$ for $|w| \ge 1$. Finally, we get,
\[ |L(w)| \le \max\{|w|^N,|w|^{-N}\}, \qquad |w| \ne 0. \]

Furthermore, we get $|v(z)| = |L(e^{iz})| \le \max\{|e^{iz}|^N,|e^{iz}|^{-N}\}$. But since $|e^{iz}| = e^{-\text{Im}(z)}$, we get $|v(z)| \le e^{N|\text{Im}(z)|}$. 

When $N=0$, $v(z)$ is a constant trigonometric polynomial and hence $v^{(r)}(0) = 0$ for $r \ge 1$. For $N>0$, using Cauchy's coefficient estimate on a circle of radius $\rho$, we get  

\[
|v^{(r)}(0)| \le \frac{r! e^{N\rho}}{\rho^r}.
\]

By taking $\rho = r/N$, we get,
 \begin{equation}\label{eq_taylor_coefficient_univariate}
    \frac{|v^{(r)}(0)|}{r!} \le \left(\frac{N e}{r}\right)^r. 
 \end{equation}

Now, let us extend the univariate estimate to $[-1,1]^n$. Take rational $y \in [-1,1]^n$ and $y_i = \frac{p_i}{q}$ where $p_i$ and $q$ are integers and $q \ge 1$ and $|p_i| \le q$. Now, we can write $R(p_1s,\dots,p_ns) = v(s)$. Taking its Taylor expansion and writing in terms of homogeneous parts, we get $R(z) = \sum_{r \ge 0} H_r(z)$. We can now write $v(s)$ in terms of homogeneous parts as
\[ v(s) = \sum_{r \ge 0} H_r(p_1s,\dots,p_ns ) = \sum_{r \ge 0} s^r H_r(p_1,\dots,p_n) . \]
For $r=0$, we get $|H_0| = |R(0)|\le 1$.
The coefficient of $s^r$ in $v(s)$ is $\frac{v^{(r)}(0)}{r!}$ and from the above equation, we get
\[ \frac{v^{(r)}(0)}{r!}  = H_r(p_1,\dots,p_n) = q^rH_r(y_1,\dots,y_n). \]

\[R(\theta) = \sum_{\alpha} c_{\alpha} e^{i\langle\alpha, \theta \rangle}, \quad \|\alpha\|_1 \le D.\]

Since $|R(\theta)| \le 1$, $|v(s) |\le 1$ for $s \in \mathbb{R}$.
\[v(s) = R(p_1s,\dots,p_ns) = \sum_{\alpha} c_{\alpha} e^{i\langle\alpha, p \rangle s}. \]
But since $|\langle\alpha, p \rangle| \le \sum_i |\alpha_i||p_i| \le \|\alpha\|_1q \le D q$.
$v$ is a trigonometric polynomial with total frequency at most $qD$.
In case of univariate trigonometric polynomials, total frequency is same as degree.
Therefore, we can use Equation~\ref{eq_taylor_coefficient_univariate} to get
\[\left| H_r(y_1,\dots,y_n) \right| = q^{-r}\left|\frac{v^{(r)}(0)}{r!}\right| \le q^{-r}\left(\frac{qD e}{r}\right)^r = \left(\frac{D e}{r}\right)^r   \]

Because $\mathbb{Q}^n \cap [-1,1]^n$ is dense in $[-1,1]^n$ and $H_r$ is continuous, we get $\left| H_r(y_1,\dots,y_n) \right|  \le \left(\frac{D e}{r}\right)^r $ for all $y \in [-1,1]^n$.
\end{proof}

\begin{lemma*}[Restatement of Lemma~\ref{lem:squarefree_approximation}]
     Let $Q: \{0,1\}^{mn} \rightarrow \mathbb{R}$ be a polynomial. Let $\{ \rho_{\theta} : \theta \in \mathbb{R}\}$ be a differentiable family of probability measures on $\{0,1\}^m$ and let $\nu_0$ and $\nu_1$ be probability measures on $\{0,1\}^m$ such that, $\nu_{(1+z)/2} = \rho_0 + \frac{z}{\kappa} \rho_0'$ for $z \in \{-1,1\}$ and some scalar $\kappa \ne 0$. Let $X^{(1)}, \dots, X^{(n)}$ be $n$ disjoint blocks of $m$ bits each and $X^{(i)}$'s are sampled independently from $\rho_{\theta_i}$. Let $R$ be defined as 
    \[ R(\theta_1,\dots,\theta_n) = \E_{X^{(i)}\sim \rho_{\theta_i}}[Q(X^{(1)}, \dots, X^{(n)})]. \]
    If $R$ is analytic in a neighborhood of 0, then, 
    \[ \SF(R)\left(\frac{z}{\kappa}\right) =  \E_{X^{(i)}\sim \nu_{\frac{1+z_i}{2}}}[Q(X^{(1)}, \dots, X^{(n)})],
    \]
    for every $z \in \{-1,1\}^n$.
\end{lemma*}
\begin{proof}
        \begin{equation}\label{R_exp_sum}
     R(\theta_1,\dots,\theta_n)= \E_{X^{(i)} \sim \rho_{\theta_i}} Q(X^{(1)},\dots,X^{(n)}) = \sum_x Q(x^{(1)},\dots,x^{(n)}) \prod_i \rho_{\theta_i}(x^{(i)}).   
    \end{equation}
    where $x=(x^{(1)},\dots,x^{(n)})$.

    For each $i \in [n]$, let $\mathcal{L}_i$ be a linear functional acting on $i^{th}$ variable such that 
    \[ \mathcal{L}_i (h) = h|_{\theta_i=0}+\frac{z_i}{\kappa} \partial_i h|_{\theta_i=0}. \]
    
    Applying the functional $\mathcal{L}_1 \dots \mathcal{L}_n$ on $R$ in Equation~\ref{R_exp_sum} gives us 
    \[ \mathcal{L}_1 \dots \mathcal{L}_n(R) = \sum_x Q(x^{(1)},\dots,x^{(n)}) \prod_i \mathcal{L}_i(\rho_{\theta_i}(x^{(i)}))  .\]

    Since $\mathcal{L}_i(\rho_{\theta_i}) = \rho_0 + \frac{z_i}{\kappa}\rho'_0$ and $\rho_0 + \frac{z_i}{\kappa}\rho'_0 = \nu_{\frac{1+z_i}{2}}$, we then get
    \[\mathcal{L}_1 \dots \mathcal{L}_n(R) = \sum_x Q(x^{(1)},\dots,x^{(n)}) \prod_i \nu_{\frac{1+z_i}{2}}(x^{(i)}). \]
    Since $\sum_x Q(x^{(1)},\dots,x^{(n)}) \prod_i \nu_{\frac{1+z_i}{2}}(x^{(i)}) = \E_{X^{(i)} \sim \nu_{\frac{1+z_i}{2}}}Q(X^{(1)},\dots,X^{(n)})$, we get 
    \[ \mathcal{L}_1 \dots \mathcal{L}_n(R)= \E_{X^{(i)} \sim \nu_{\frac{1+z_i}{2}}}Q(X^{(1)},\dots,X^{(n)}).\]
    On the other hand, $R$ is analytic in a neighborhood of 0 and hence can be written as

    \[ R(\theta)= \sum_{\alpha \in \mathbb{Z}_{\ge 0}^n} \frac{\partial^{\alpha}R(0)}{\alpha!} \prod_i {\theta}_i^{\alpha_i}, \]
where $\alpha! = \alpha_1 ! \dots \alpha_n!$ and $\partial^{\alpha}={\partial_1}^{\alpha_1}\dots {\partial_n}^{\alpha_n} $.

Applying $\mathcal{L}_1 \dots \mathcal{L}_n$ to $R$, we get $\mathcal{L}_1 \dots \mathcal{L}_n(R)= \sum_{\alpha \in \mathbb{Z}_{\ge 0}^n} \frac{\partial^{\alpha}R(0)}{\alpha!} \prod_i \mathcal{L}_i({\theta}_i^{\alpha_i})$.

\[ \mathcal{L}_i({\theta}_i^{\alpha_i}) = \begin{cases}
    1  , \qquad &\alpha_i=0, \\  
    \frac{z_i}{\kappa} \alpha_i {\theta}_i^{\alpha_i-1}|_{\theta_i=0},  &\alpha_i \ge 1.
\end{cases}  \]
For all the terms where $\alpha_i \ge 2$, this evaluates to 0 and so the only terms that remain are $(\frac{z_i}{\kappa})^{\alpha_i}$.

\[\mathcal{L}_1 \dots \mathcal{L}_n(R)= \sum_{\alpha \in \{0,1\}^n } \partial^{\alpha}R(0) \prod_{i} \left(\frac{z_i}{\kappa}\right)^{\alpha_i}\]

This is the squarefree extraction of $R$ evaluated at $(z/\kappa)$.

\[ \SF(R)\left(\frac{z}{\kappa}\right) = \E_{X^{(i)} \sim \nu_{(1+z_i)/2}}Q(X^{(1)},\dots,X^{(n)}) .\]

\end{proof}

\section{Proofs of Lemmas required for the Main Result}\label{appendix-proof-main-result}

\begin{lemma*}[Restatement of Lemma~\ref{lemma_def_deg_operator}]
    Let $g:\{0,1\}^m \rightarrow \{0,1\}$ be a non-constant Boolean function with $\adeg(g)=d$. Then there exists a non-empty set $\mathcal{X} \subseteq \{0,1\}^m$, real orthogonal projectors $E_0,E_1,\dots,E_m$ on $\mathbb{R}^{\mathcal{X}}$ and a real symmetric matrix $A = \sum_{k=0}^m \tau_k E_k$ where $\tau_k= \min\{k,d\}$ such that they satisfy the following three properties:
    \begin{align}
        &\sum_{k=0}^m E_k = I, \qquad 0 \preceq A \preceq dI;  \\
        &E_kM_pE_{\ell} = 0 \quad \text{if} \; |k-\ell|>r \quad \text{for every polynomial of degree at most} \; r;  \\
        & \frac{d}{3} \le \lambda = \|P_0AP_1\| \le \frac{d}{2},  \; \text{where} \; P_b = M_{\mathbf{1}_{\{g=b\}}}.
    \end{align}
\end{lemma*}

\begin{proof}
    Let $\psi$ be the dual witness of $g$ described in the Fact~\ref{dual_witness_properties}. Let $\mathcal{X} = \{x: |\psi(x)| >0 \}$.
    Since $\|\psi\|_1 = 1$, $\mathcal{X}$ is non-empty.
    Let $\delta$ be a vector such that $\delta_x = \sqrt{|\psi(x)|}$. Also, on $x \in \mathcal{X}$, define $w_x = \frac{\psi(x)}{\sqrt{|\psi(x)|}}$. We can define real subspaces as follows:
    \[ V_{-1} = \{0\}, \qquad V_k = \{M_p \delta: \deg p \le k \} \; (0 \le k \le m).  \]

    Notice that $V_0 \subseteq V_1 \subseteq \dots \subseteq V_m$. Let $W_k = V_k \cap V_{k-1}^{\perp}$. For all $0 \le k \le m$, $V_k$ can thus be written as 
    \[ V_k = W_0 \oplus \dots \oplus W_k.  \]
    Also, $W_i \perp W_j$ for $i \ne j$. Therefore, $V_m = W_0 \oplus \dots \oplus W_m$ is an orthogonal direct sum.
    Let $E_k$ be an orthogonal projector onto $W_k$. Let $e_x$ be a standard basis vector.

    Let $x = (x_1,\dots,x_m) \in \mathcal{X}$. Let $p_x(y) = \prod_{j:x_j=1} y_j \prod_{j:x_j=0} (1-y_j)$. Then $p_x(y) = \mathbf{1}_{x=y}$ and $\deg(p_x) \le m$. Also, $M_{p_x} \delta = \delta_x e_x$. Since $x \in \mathcal{X}$, $\delta_x >0$. Therefore, 
    \[e_x = \frac{M_{p_x}\delta}{\delta_x} \in V_m.\]
    Thus, every standard basis vector of $\mathbb{R}^{\mathcal{X}}$ lies in $V_m$.
    Therefore, $V_m = \mathbb{R}^{\mathcal{X}}$ and consequently, 
    \begin{equation}\label{sum_orthogonal_projectors}
      \sum_k E_k = I.  
    \end{equation}

    Let $A = \sum_k \tau_k E_k$ where $\tau_k = \min\{k,d\}$. $A$ is real and symmetric. Let $v \in W_j$ for some $0\le j \le m$. Then $Av = \sum_k \tau_k E_k v = \tau_j v$. Therefore, $v$ is an eigenvector and $\tau_j$ is the corresponding eigenvalue. Thus, $A = \sum_k \tau_k E_k$ is its spectral decomposition and its eigenvalues lie in $[0,d]$. Hence we get $0 \preceq A \preceq dI$.

    Now, let $p$ be a polynomial with degree $r$. First suppose $\ell+ r < k  \le m$. Then, $ M_p V_{\ell} \subseteq V_{\ell +r}$. If we take a vector $v$, then $E_{\ell} v \in W_{\ell} \subseteq V_{\ell}$. Then $M_p E_{\ell} v \in V_{\ell + r}$. Then $E_k M_p E_{\ell}v=0$ because $E_k$ is an orthogonal projector onto $W_k$ and $W_k \subseteq V_{k-1}^{\perp}$ and $V_{\ell+r} \subseteq V_{k-1}$. 
    By the same argument if we interchange $k$ and $\ell$, $E_{\ell} M_p E_{k} = 0$ if $\ell > r+k$. Since $E_k$,$E_{\ell}$ and $M_p$ are real and symmetric, we get $(E_{\ell} M_p E_{k})^T = E_k M_p E_{\ell}=0$. Therefore, 
    \[ E_kM_pE_{\ell} = 0 \quad \text{whenever} \; |k-\ell|>r. \]

    Let $v \in \mathbb{R}^\mathcal{X}$. Then $(P_b v)_x = v_x$ if $x \in \mathcal{X} \cap g^{-1}(b)$, otherwise, it is 0. Since $g^{-1}(0)$ and $g^{-1}(1)$ are disjoint, $P_0P_1 = 0$. Therefore, we can write $P_0AP_1 = P_0AP_1 - \frac{d}{2}P_0IP_1 $. Let $B = \left(A - \frac{d}{2}I\right)$. 
    Since $A$ is symmetric, $B$ is symmetric. Also, eigenvalues of $B$ lie in $[-d/2,d/2]$. Therefore, $\|B\| \le d/2$. Also, since $P_0$ and $P_1$ are orthogonal projectors, $\|P_0\|=\|P_1\|=1$
    
    Thus, $\lambda= \|P_0 A P_1\| = \|P_0BP_1\| \le \|P_0\| \; \|B\| \; \|P_1\| \le d/2$. 

    Now, note that $\delta \in V_0$. This is because if you take $M_p$ corresponding to constant polynomial $p=1$, then $M_p \delta = \delta$. Since $p$ has degree $0$, $M_p \delta = \delta \in V_0$. Also, since $W_0=V_0$, $\delta \in W_0$ and hence $E_0 \delta = \delta$ but $E_k \delta =0$ for all $k \ge 1$. Therefore, $A \delta = 0$.

     For a polynomial with degree at most $d-1$, $M_p \delta \in V_{d-1}$. Therefore, using the orthogonality condition in Equation~\ref{high_pure_degree} in the second equality below, we get,
    \[\langle w,M_p \delta \rangle  = \sum_{x \in \mathcal{X}} w_x p(x)\delta_x = \sum_{x \in \mathcal{X}} \psi(x)p(x) = 0. \]

    Therefore, $w \perp V_{d-1}$. We know $V_{d-1} = W_0 \oplus \dots \oplus W_{d-1}$. Therefore, $E_k w = 0$ for $0 \le k \le d-1 $. Since $\tau_k = \min\{k,d\}$, $Aw = \sum_{k=0}^m \tau_k E_k w = \sum_{k=d}^m d E_k w$. Also, $\sum_{k=0}^m E_k w = w$ from Equation~\ref{sum_orthogonal_projectors} and $\sum_{k=0}^{d-1} E_kw = 0$, we get $Aw = dw$. 

    Now, by Cauchy-Schwarz we have $\|w\|_2 \; \|[A,P_1]\| \|\delta\|_2 \ge |\langle w,[A,P_1] \delta \rangle|$ , but since $\|w\|=1$ and $\|\delta\|=1$ from $\|\psi\|_1 = 1$,
    \begin{align*}
        \|[A,P_1]\| &\ge |\langle w,[A,P_1] \delta \rangle| \\
                    &= |\langle w,AP_1 \delta \rangle - \langle w,P_1A \delta \rangle| \\
                    &= |\langle Aw,P_1 \delta \rangle| \\
                    &=|\langle dw,P_1 \delta \rangle|\\
    \end{align*}
    
    The second equality uses two facts: $A\delta =0$ and $A = A^T$. Now, we know $\langle dw,P_1 \delta \rangle = d\sum_{x \in \mathcal{X} \cap g^{-1}(1)} w_x \delta_x = d\sum_{x \in \mathcal{X} \cap g^{-1}(1)} \psi(x) = d\sum_{x \in \mathcal{X} } \psi(x)g(x)$. From Equation~\ref{correlation}, we have $\sum_{x \in \mathcal{X} } \psi(x)g(x) \ge 1/3$ and hence we get $\|[A,P_1]\| \ge d/3$. 
    Now, since $P_0P_1=0$ and $P_0+P_1=I$, we can write $\mathbb{R}^{\mathcal{X}} = P_0\mathbb{R}^{\mathcal{X}} \oplus P_1 \mathbb{R}^{\mathcal{X}}$. Then, we can write the matrix $A$ as 
    \[ A = \begin{pmatrix}
        P_0AP_0 & P_0AP_1 \\
        P_1AP_0 & P_1AP_1
    \end{pmatrix}. \]
    Also, matrix $P_1$ can be written as 
    \[ P_1 = \begin{pmatrix}
        P_0P_1P_0 & P_0P_1P_1 \\
        P_1P_1P_0 & P_1P_1P_1
    \end{pmatrix} 
 = \begin{pmatrix}
        0 & 0 \\
        0 & P_1
    \end{pmatrix}. \]

    Therefore, 
    \[ [A,P_1] = AP_1 - P_1A = \begin{pmatrix}
        0 & P_0AP_1 \\
        -P_1AP_0 & 0
    \end{pmatrix}. \]
    Let $C=P_0AP_1$ and because $P_1,P_0,A$ are real symmetric, $C^T = (P_0AP_1)^T =P_1AP_0$. Therefore, 
    \[[A,P_1]^T [A,P_1] = \begin{pmatrix}
        0 & -C \\
        C^T & 0
    \end{pmatrix} \begin{pmatrix}
        0 & C \\
        -C^T & 0
    \end{pmatrix} = \begin{pmatrix}
        CC^T & 0 \\
        0 & C^TC
    \end{pmatrix}\]

    $\left\lVert[A,P_1]^T [A,P_1] \right\rVert = \|[A,P_1]\|^2$ and $\|[A,P_1]^T [A,P_1]\| = \max\{\|CC^T\|,\|C^TC\| \}= \|C\|^2 = \|P_0AP_1\|^2$.
    
    Hence we get $\|[A,P_1]\| = \|P_0AP_1\|$ and hence $\lambda =\|P_0AP_1\| \ge d/3$.
\end{proof}

\begin{lemma*}[Restatement of Lemma~\ref{prob_inner_arbitrary}]
    Let $\mathcal{X},A, \lambda,P_0,P_1$ be as defined in previous Lemma~\ref{lemma_def_deg_operator} and $\kappa = 2 \lambda$. There exist real unit vectors $u_0 \in P_0 \mathbb{R}^{\mathcal{X}}$ and $u_1\in P_1 \mathbb{R}^{\mathcal{X}}$ such that
     \[ P_0 A u_1 = \lambda u_0, \qquad  P_1 A u_0 = \lambda u_1. \]
    Define $v = \frac{u_0+i u_1}{\sqrt{2}}$ and for real $\theta$, define $\rho_{\theta}: \mathcal{X} \rightarrow [0,1]$ by
    \[\rho_{\theta}(x) = \left| (e^{i \theta A} v)_x \right|^2.\]
    Then $\{\rho_{\theta}: \theta \in \mathbb{R}\}$ is a family of probability measures.
    
    Also, there are probability measures $\nu_0$ and $\nu_1$ supported on $\mathcal{X} \cap g^{-1}(0)$ and $\mathcal{X} \cap g^{-1}(1)$ respectively, such that
    \[ \rho_0 = \frac{\nu_0 + \nu_1}{2}, \qquad \rho_0' =\frac{\kappa}{2} (\nu_1 - \nu_0). \]

    Consequently, for $z \in \{-1,1\}$
    \[ \nu_{\frac{1+z}{2}} = \rho_0 + \frac{z}{\kappa} \rho_0' . \]
\end{lemma*}

\begin{proof}
    We know from previous Lemma~\ref{lemma_def_deg_operator}, $\lambda = \|P_0 A P_1\| \ge d/3 > 0$. Since this is in  finite dimensions, $\|P_0 A P_1\|$ is the largest singular value of $P_0 A P_1$. We can therefore find a pair of real unit vectors $u_0$ and $u_1$ such that 
    \[ P_0 A P_1 u_1 = \lambda u_0, \qquad  (P_0 A P_1)^*u_0 = \lambda u_1. \]

    Since $P_0,A,P_1$ are real symmetric, $(P_0 A P_1)^* = P_1 AP_0$. Also, since $\lambda > 0$, the singular-vector equations imply that $u_1$ lies in the range of $P_1$ and $u_0$ lies in the range of $P_0$ and hence, $P_1 u_1 = u_1$ and $P_0 u_0 = u_0$.

    \[ P_0 A u_1 = \lambda u_0, \qquad  P_1 A u_0 = \lambda u_1. \]

    Let $v = \frac{u_0+i u_1}{\sqrt{2}}$. $v$ is a unit vector and so $\|v\|_2=1$.
     
    Also, for real $\theta$, let $\rho_{\theta}: \mathcal{X} \rightarrow [0,1]$ such that $\rho_{\theta}(x) = \left| (e^{i \theta A} v)_x \right|^2$.
    Since $A$ is symmetric and therefore Hermitian, $e^{i \theta A}$ is unitary. Therefore, 
    \[ \sum_{x \in \mathcal{X}}\rho_{\theta}(x) = \sum_{x \in \mathcal{X}} \left| (e^{i \theta A} v)_x \right|^2 = \|e^{i \theta A} v\|_2^2 = \|v\|_2^2 =1\]
    Thus, $\rho_{\theta}$ is a probability measure for real $\theta$.
    
    Let $\mathcal{X}_b = \mathcal{X} \cap g^{-1}(b)$ for $b \in \{0,1\}$ and $\nu_b(x) = u_b(x)^2$ for $x \in \mathcal{X}$. $\nu_b(x) \ge 0$ and $\sum_{x \in \mathcal{X}_b} \nu_b(x) = \|u_b\|_2^2 = 1$. Thus, $\nu_0$ and $\nu_1$ are probability measures supported on $\mathcal{X}_0$ and $\mathcal{X}_1$ respectively.

    \[ \rho_0(x) =\left| (e^{i 0 A} v)_x \right|^2 = |v_x|^2.  \]

    For every $x \in \mathcal{X}$, either $\nu_1(x) = 0$ or $\nu_0(x) = 0$. Therefore, $\rho_0(x) = \frac{\nu_0(x)+ \nu_1(x)}{2}$ for $x \in \mathcal{X}$. 
    Let $w(\theta) = e^{i \theta A} v$ and so $\rho_{\theta}(x) = |w_x(\theta)|^2 = w_x(\theta) \overline{w_x(\theta)}$.
    \[ \rho_{\theta}'(x) =  w'_x(\theta) \overline{w_x(\theta)} + w_x(\theta) \overline{w'_x(\theta)} =\overline{ w_x(\theta) \overline{w'_x(\theta)}}+ w_x(\theta) \overline{w'_x(\theta)} = 2 \text{Re}(w'_x(\theta) \overline{w_x(\theta)}). \]

    Since, $\overline{w_x(0)} = \overline{v_x}$ and $w_x'(0)=(iAv)_x$, \[\rho_{0}'(x) = 2 \text{Re}(\overline{v_x}(iAv)_x).\]
    For $x \in \mathcal{X}_0$, $v_x = \frac{(u_0)_x}{\sqrt{2}}$. Similarly for $x \in \mathcal{X}_1$, $v_x = \frac{i(u_1)_x}{\sqrt{2}}$. 

    We have \[(iAv)_x = \frac{(iAu_0)_x-(Au_1)_x}{\sqrt{2}}.\]
    Let us first see $x \in \mathcal{X}_0$. \[\overline{v}_x(iAv)_x = \frac{(u_0)_x ((iAu_0)_x-(Au_1)_x) }{2}.\] The real part is $-\frac{(u_0)_x(Au_1)_x}{2}$. But $(Au_1)_x = (P_0Au_1)_x = \lambda u_0(x) $. Hence, $\rho'_0(x) = -\lambda u_0(x)^2$.

    For $x \in \mathcal{X}_1$, \[\overline{v_x}(iAv)_x = \frac{-i(u_1)_x ((iAu_0)_x-(Au_1)_x) }{2}.\] The real part is $\frac{(u_1)_x(Au_0)_x}{2}$. But $(Au_0)_x = (P_1Au_0)_x = \lambda u_1(x) $. Hence, $\rho'_0(x) = \lambda u_1(x)^2$.

    Therefore, we get $\rho'_0 = \frac{\kappa}{2}( \nu_1 - \nu_0)$.
\end{proof}

\begin{lemma*}[Restatement of Lemma~\ref{squarefree_probability_outputclass_relation}]
    Let $Q$ be a real multilinear polynomial on $n$ disjoint block of $m$ bits each. Let degree of $Q$ be at most $D$ and $0 \le Q \le 1$ on the Boolean cube $\{0,1\}^{mn}$. Let $X^{(1)},\dots, X^{(n)}$ be $n$ independent disjoint blocks such that $X^{(i)}$ is sampled from $\rho_{\theta_i}$, where $\theta_i \in \mathbb{R}$ and $\rho_{\theta_i}$ are as described in the previous Lemma~\ref{prob_inner_arbitrary}. Let $R$ be defined as follows: 
    \[ R(\theta_1,\dots,\theta_n) = \E_{X^{(i)}\sim \rho_{\theta_i}}[{Q}(X^{(1)},\dots,X^{(n)})]. \]

    Then, $R$ is a real-valued trigonometric polynomial with total frequency at most $D$ and $0 \le R \le 1$ on $\mathbb{R}^n$. Furthermore, the squarefree extraction of $R$ satisfies:
    \begin{equation}
        \SF(R)(z/\kappa) = \E_{X^{(i)}\sim \nu_{(1+z_i)/2}} Q(X^{(1)},\dots,X^{(n)})  \qquad (z \in \{-1,1\}^n).
    \end{equation}
\end{lemma*}

\begin{proof}
    $R$ is bounded on $\mathbb{R}^n$ because $0\le Q \le 1$ on Boolean cube $\{0,1\}^{mn}$ and $X^{(i)} \in \mathcal{X} \subseteq \{0,1\}^m$. Also $R$ is real-valued because $Q$ is real for every $X^{(i)} \in \mathcal{X} \subseteq \{0,1\}^m$. Let $M$ be a monomial of $Q$ and let $S$ be the exact set of indices of blocks containing a variable of $M$. We can factor $M$ into $M^{i}$ where $M^{i}$ is a monomial with variables of $M$ from $i^{th}$ block. $\E[M] = \E\left[\prod_{i \in S} M^{i}\right]$. Since $X^{(i)}$ blocks are independent, we can write  
    \[\E[M] = \E\left[\prod_{i \in S} M^{i}\right] = \prod_{i \in S} \E[M^i].\]

    Now let us look at $\E[M^i]$. Let $p=M^i$ and let $r_i = \deg M^i$. We can write it as $\E[p] = \sum_{x \in \mathcal{X}} p(x) \rho_{\theta_i}(x) $ . Therefore,
    \[\E[p] = \sum_{x \in \mathcal{X}} p(x) |(e^{i \theta_i A}v)_x|^2. \]
    Let $e^{i \theta_i A}v = u$. Then, $\E[p] = \sum_{x \in \mathcal{X}} p(x) |u_x|^2 $. Since $(M_p u)_x = p(x)u_x$, $\sum_{x \in \mathcal{X} } p(x) |u_x|^2 = \langle u,M_p u \rangle$. Therefore,
    \[\E[p] = \langle e^{i \theta_i A}v,M_p e^{i \theta_i A}v \rangle.\]
    Since $A = \sum_k \tau_k E_k $ by spectral decomposition and $E_k$'s are orthogonal projectors, $e^{i \theta_i A} = \sum_k e^{i \theta_i \tau_k} E_k$.
    \[\E[p] = 
    \left\langle \sum_k e^{i \theta_i \tau_k} E_k v,M_p \sum_{\ell} e^{i \theta_i \tau_{\ell}} E_{\ell}v 
    \right\rangle = \sum_{k,\ell} e^{i \theta_i (-\tau_k)} e^{i \theta_i \tau_{\ell}} \left\langle   E_k v,M_p  E_{\ell}v 
    \right\rangle = \sum_{k,\ell} e^{i \theta_i (\tau_{\ell}-\tau_k)}     v^*E_k M_p  E_{\ell}v 
    . \]
    Since $p$ is of degree $r_i$, by using Lemma~\ref{lemma_def_deg_operator}, we get $E_k M_p  E_{\ell} = 0$ for $|k-\ell| > r_i$.   
    Also, $\E[p]$ has frequencies of the form $\tau_{\ell}-\tau_{k}$. Since $|\tau_k-\tau_{\ell}| = |\min\{k,d\}-\min\{\ell,d\}| \le |k - \ell|$ and the terms vanish for $|k - \ell| >r_i$, $\E[p]$ has frequency at most $r_i$. Let $\sum_{\alpha} c_{\alpha} e^{i\langle \alpha,\theta \rangle}$ be the trigonometric polynomial for monomial $M$ where $\alpha_i$ comes from individual monomial $M^i$. Then, since $\deg(M^i) = r_i \ge 0$, $|\alpha_i| \le r_i$ and since $\sum_i r_i \le D$, the total frequency for every monomial is at most $D$.
     By linearity of expectation, $R$ has total frequency at most $D$.

    Also, since $R$ is trigonometric, it is analytic in a neighborhood of 0. Also $\nu_{(1+a)/2} = \rho_0 + \frac{a}{\kappa}\rho_0'$ for $a \in \{-1,1\}$. Applying Lemma~\ref{lem:squarefree_approximation},
    \[ \SF(R)\left(\frac{z}{\kappa}\right) = \E_{X^{(i)} \sim \nu_{(1+z_i)/2}}Q(X^{(1)},\dots,X^{(n)}) .\]

\end{proof}

\section{Proof of Warm up Lemmas}
\subsection{Proofs of Lemmas for function composed with OR}\label[appendix]{appendix-proof-f-or}

\begin{lemma*}[Restatement of Lemma~\ref{lem_or_dilation}]
    Let $f: \{0,1\}^n \rightarrow \{0,1\}$ be a non-constant Boolean function with $\adeg(f)=d_f \ge 1$. Let $R : \mathbb{R}^n \rightarrow \mathbb{R}$ be a polynomial with degree at most $d$. Additionally, $R(x)\in [0,1]$ for all $x \in [0,1]^n$. If for some $a \in [0,1]$,
    \[|R(ay) - f(y)| \le \frac{1}{10},\]
    for all $y \in \{0,1\}^n$, then 
    \[d \sqrt{a} > \frac{d_f}{4e}.\]
\end{lemma*}

\begin{proof}
    We cannot have $d < d_f$ because then $h(y)=R(ay)$ will approximate $f(y)$ with error at most $1/10$ and then $\adeg(f) \le d < d_f $ which is a contradiction and hence, $d \ge d_f$. We can write $R$ as sum of its homogeneous parts, that is, $R=\sum_{k=0}^{d} R_k$ where $R_k$ is a homogeneous polynomial of degree $k$. 

    Fix an arbitrary $y \in \{0,1\}^n$, define $g_y: \mathbb{R} \rightarrow \mathbb{R}$ such that $g_y(t)=R(ty)$. We have $\deg(g_y) \le d$ and $g_y(t)$ can be written as $g_y(t)=\sum_k t^k R_k(y)$ . For $t \in [0,1]$, $g_y(t)=R(ty) \in [0,1]$. Also, $g_y^{(k)}(0) = k! R_k(y)$.

    Using Markov's higher-derivative inequality for $k \le d$ on $[0,1]$ gives,

    \[ |g_y^{(k)}(0)| \le 2^k\frac{(d^2)(d^2-1^2)\dots(d^2-(k-1)^2)}{(2k-1)!!} \le \frac{2^kd^{2k}}{(2k-1)!!}. \]

    Using the identity $2^kk! (2k-1)!! = (2k)!$, we get, 
    \begin{equation}\label{r_k_upperbound}
        |R_k(y)| = \frac{|g_y^{(k)}(0)|}{k!} \le \frac{(2d)^{2k}}{(2k)!}.
    \end{equation}
    Since $y$ was arbitrary, this bound holds for all $y \in \{0,1\}^n$.
     Let $S$ be such that
    \[S(y) = \sum_{k=0}^{d_f-1} a^kR_k(y).\]

    Since $\deg S < d_f$, we can write $|R(ay)-S(y)| = |\sum_{k=0}^{d} a^kR_k(y) - \sum_{k=0}^{d_f-1} a^kR_k(y)| = |\sum_{k=d_f}^{d} a^kR_k(y) |\le  \sum_{k=d_f}^{d} a^k|R_k(y)|$. Using Equation~\ref{r_k_upperbound},
    \[|R(ay)-S(y)| \le  \sum_{k=d_f}^{d} a^k|R_k(y)|  \le \sum_{k=d_f}^{d} a^k\frac{(2d)^{2k}}{(2k)!} = \sum_{k=d_f}^{d} \frac{2^{2k}(d \sqrt{a})^{2k}}{(2k)!} .\]
    Now, assuming $d \sqrt{a} \le \frac{d_f}{4e}$, we get
    \[ |R(ay)-S(y)| \le \sum_{k=d_f}^{d} \frac{(2d_f)^{2k}}{(4e)^{2k} (2k)!} \le \sum_{k=d_f}^{d} \frac{(2k)^{2k}}{(4e)^{2k} (2k)!} . \]
    Since $e^{2k}(2k)! \ge (2k)^{2k}$, we get,
    \[ |R(ay)-S(y)| \le \sum_{k=d_f}^{d} \frac{1}{16^k}  \le \sum_{k=d_f}^{\infty} \frac{1}{16^k}  \le \frac{1}{15}. \]

    Therefore, $|S(y)-f(y)| \le |S(y)-R(ay)|+|R(ay)-f(y)| \le 1/10+ 1/15 < 1/3$ for all $y \in \{0,1\}^n$. The assumption $d \sqrt{a} \le \frac{d_f}{4e}$ led to a contradiction that a polynomial $S$ with degree less than $d_f$ approximates $f$ on the Boolean cube. Hence, $d \sqrt{a} > \frac{d_f}{4e}$.
    \end{proof}

\begin{lemma*}[Restatement of Lemma~\ref{lem:bernoulli_avg}]
    Let $Q:\{0,1\}^{mn} \rightarrow \mathbb{R}$ be a multilinear polynomial with degree at most $d$. Let $X=(X^{(1)},\dots,X^{(n)})$ be $n$ disjoint blocks of $m$ bits each. Let $p=(p_1,\dots,p_n) \in [0,1]^n$. For all $mn$ bits, let $X^{(i)}_j$ be sampled independently from $\operatorname{Bernoulli}(p_i)$. Then if $R$ is defined as
    \[ R(p_1,\dots,p_n) = \E[Q(X^{(1)},\dots,X^{(n)})], \]
    then $\deg(R) \le d$. Furthermore, if $0 \le Q \le 1$ for all $x \in \{0,1\}^{mn}$, then $0 \le R \le 1$ for $p \in [0,1]^n$. 
\end{lemma*}

\begin{proof}
    Let $M^S$ be a monomial of $Q$ where $S \subseteq [n]$ is the set of blocks from which $M^S$ contains at least one bit. It can be factored into monomials $M^i$ such that the monomial $M^i$ only contains bits from $i^{th}$ block. 
        \[M^S = \prod_{i \in S} M^i. \]
    Additionally, the blocks are disjoint and each bit in one block is sampled independently. Let $\deg(M^i) =r_i$. Therefore,
        \[\E[M^S] = \E\left[\prod_{i\in S} M^i\right] = \prod_{i \in S} \E [M^i] = \prod_{i \in S} (p_i)^{r_i}. \]
    Since $\sum_{i \in S} r_i = \deg(M^S) \le d$, we get degree of $\E [M^S]$ is at most $d$. By linearity of expectation, we get degree of $R=\E Q $ is at most $d$. If $0 \le Q \le 1$, then since $R$ is expectation of Q with respect to product Bernoulli distribution, $0 \le R(p_1,\dots,p_n)\le 1$ for $p=(p_1,\dots,p_n) \in [0,1]^n$.
\end{proof}


\subsection{Proof of Lemmas for Lower Bound for function composed with promise threshold} 
\label{appendix:exact_threshold}

\begin{lemma*}[Restatement of Lemma~\ref{lem:prob_dist_th}]
    Let $m \ge 1$ and $0 \le k \le m-1$. Let $\kappa = \sqrt{(k+1)(m-k)}$. Let $\sigma_k = \operatorname{Uniform}\{x : |x|=k\}$, $\sigma_{k+1} = \operatorname{Uniform}\{x : |x|=k+1\}$ and let $\nu_0 = \sigma_k$ and $\nu_1 = \sigma_{k+1}$. 
    Define vectors $v_j$ for $j \in \{k,k+1\}$ and $\phi$ in $\mathbb{R}^{\{0,1\}^m}$ by
    \[ v_j = \binom{m}{j}^{-1/2} \sum_{|x|=j} e_x \;, \quad  \phi = \frac{v_k + v_{k+1}}{\sqrt{2}}. \] 
    Let $U_{\theta}$ be as defined in Equation~\ref{eq:rotation_matrix}. Define $\rho_{\theta}: \{0,1\}^m \rightarrow [0,1]$ by
    \[\rho_{\theta}(x) =  (\langle e_x, U_{\theta}^{\otimes m} \phi \rangle)^2.\]
    Then $\{\rho_{\theta}: \theta \in \mathbb{R}\}$ is a family of probability measures. Furthermore,
    \[ \rho_0 = \frac{\sigma_k + \sigma_{k+1}}{2}, \quad \rho_0' = \frac{\kappa}{2}(\sigma_{k+1}-\sigma_{k}). \]
    Consequently, for $z \in \{-1,1\}$, $ \rho_0 + \frac{z}{\kappa} \rho_0' = \sigma_{k+(1+z)/2} = \nu_{(1+z)/2}$.
\end{lemma*}
\begin{proof}
    
    Since $\sigma_k$ is uniform distribution on $\{x \in \{0,1\}^m:|x|=k\}$,
    \begin{equation}
        \sigma_k(x)=\begin{cases}
            \frac{1}{\binom{m}{k}}, \quad &|x|=k , \\
             \; 0, \quad &|x| \ne k.
        \end{cases}
    \end{equation}

    From the definition of $v_k$, we have
    \begin{equation}
        v_k(x)=\begin{cases}
            \sqrt{\frac{1}{\binom{m}{k}}}, \quad &|x|=k , \\
            \; 0, \quad &|x| \ne k.
        \end{cases}
    \end{equation}
    
    Let us see that $\rho_{\theta}$ is a family of probability measures. Clearly, $\rho_{\theta}(x) \ge 0$. Also,
    \[ \sum_{x\in \{0,1\}^m} \rho_{\theta}(x) = \sum_{x\in \{0,1\}^m}  ((U_{\theta}^{\otimes m} \phi)_x)^2 = \|U_{\theta}^{\otimes m} \phi\|^2_2  \]
    We know that $v_k$ and $v_{k+1}$ are unit vectors and hence $\|v_k\|_2=1$ and $\|v_{k+1}\|_2= 1$ and since they are orthogonal, $\phi$ is a unit vector and thus, $\|\phi\|_2=1$. Also, $U_{\theta}$ is orthogonal and hence $\|U^{\otimes m}_{\theta} \phi\|_2 = \|\phi\|_2$. Finally, $\sum_{x\in \{0,1\}^m} \rho_{\theta}(x) =1$  

    Additionally, \[\rho_0(x) = ((U_{0}^{\otimes m} \phi)_x)^2 = ((\phi)_x)^2 = \frac{({v_k}_{x}+{v_{k+1}}_{x})^2}{2} = \frac{\sigma_k(x)+\sigma_{k+1}(x)}{2}. \]

    \[ \rho'_{\theta}(x) = 2 (U_{\theta}^{\otimes m} \phi)_x\left(\frac{d}{d \theta}U_{\theta}^{\otimes m} \phi \right)_x, \qquad \rho'_{0}(x) = 2 ( \phi)_x\left(\left.\frac{d}{d \theta}U_{\theta}^{\otimes m}\right|_{\theta=0} \phi \right)_x.\]

    Then at $\theta=0$, $U_{\theta}^{\otimes m}=I$,

        \[\frac{d}{d \theta}U_{\theta}^{\otimes m} = \sum_{j=1}^m U_{\theta}^{\otimes (j-1)}\otimes U_{\theta}'\otimes U_{\theta}^{\otimes (m-j)},  \qquad \left.\frac{d}{d \theta}U_{\theta}^{\otimes m}\right|_{\theta=0} = \sum_{j=1}^m I^{\otimes (j-1)}\otimes U_{0}'\otimes I^{\otimes (m-j)}.  \]

    Let $E_j=I^{\otimes (j-1)}\otimes U_{0}'\otimes I^{\otimes (m-j)}$.
    Now, \[ U_{\theta}'= \frac{d}{d \theta}U_{\theta} = \frac{1}{2} \begin{pmatrix} 
    -\sin{(\theta/2)} & -\cos{(\theta/2)} \\
    \cos{(\theta/2)} & -\sin{(\theta/2)}
    \end{pmatrix} \qquad U_{0}'=
    \left.\frac{d}{d \theta}U_{\theta}\right|_{\theta=0} =\frac{1}{2}  \begin{pmatrix}
    0 & -1 \\
    1 & 0
    \end{pmatrix}.
    \] 
    Hence, 
    \[\rho'_{0}(x) = 2 ( \phi)_x\left(\sum_j E_j \phi \right)_x = 2(\phi)_x \left(\sum_j \frac{E_j v_k +E_j v_{k+1}}{\sqrt{2}}\right)_x.\]
    Using $U'_0 e_0 = \frac{1}{2} e_1$ and $U'_0 e_1 = -\frac{1}{2} e_0$, we get
    \begin{align*}
        E_j v_k &= \frac{1}{\sqrt{\binom{m}{k}}} \sum_{|y|=k} E_j e_y = \frac{1}{\sqrt{\binom{m}{k}}} \sum_{|y|=k} \bigotimes_{i=1}^{j-1} e_{y_i} \otimes (-1)^{y_j} \frac{1}{2} e_{1-y_j} \otimes \bigotimes_{i=j+1}^m e_{y_i} \\
        &= \frac{1}{2\sqrt{\binom{m}{k}}}\left( \sum_{|y|=k+1,y_j=1} e_y - \sum_{|y|=k-1,y_j=0} e_y \right).    
    \end{align*}
    
    Similarly, 
     \begin{align*} E_j v_{k+1} = \frac{1}{2\sqrt{\binom{m}{k+1}}}\left( \sum_{|y|=k+2,y_j=1} e_y - \sum_{|y|=k,y_j=0} e_y \right).
     \end{align*}
     For $x$ with Hamming weights other than $k$ and $k+1$, $\phi_x = 0$. Therefore, we will only look at $x$ where $|x|=k$ or $|x|=k+1$.
     Let us take $x$ with $|x|=k$. $(\phi)_x = \frac{1}{\sqrt{2 \binom{m}{k}}}$ and $(E_j v_k)_x=0$. For $(E_j v_{k+1})$, only the second sum contributes to the total sum. Therefore, for $x$ with $|x|=k$ and $x_j=0$, 
      \[(E_j v_{k+1})_x = -\frac{1}{2\sqrt{\binom{m}{k+1}}}. \]
     Since, $x$ has $m-k$ zeros, there are only $m-k$ possibilities of $j$ that contribute to the final sum.
     Therefore,
     \[\rho'_{0}(x) = 2(\phi)_x \left(\sum_j \frac{E_j v_k +E_j v_{k+1}}{\sqrt{2}}\right)_x = 2 \frac{1}{\sqrt{2 \binom{m}{k}}}\left( -\frac{1}{2\sqrt{2\binom{m}{k+1}}} (m-k) \right) = \left( \frac{-(m-k)}{2 \sqrt{\binom{m}{k+1} \binom{m}{k}}}   \right).\]

     Since,
     \[ \sqrt{\binom{m}{k+1} \binom{m}{k}} = \sqrt{\frac{(m-k)}{k+1}} \binom{m}{k},  \]
     we get $\rho'_{0}(x) = -\frac{\sqrt{(k+1)(m-k)}}{2\binom{m}{k}} = -\frac{\kappa}{2} \sigma_k(x)$. Following similar calculation for $x$ where $|x|=k+1$, we get $\rho'_{0}(x) = \frac{\kappa}{2} \sigma_{k+1}(x)$.
     
\end{proof}

\begin{lemma*}[Restatement of Lemma~\ref{lem:threshold_trig_expectation}]
     Let $\phi \in \mathbb{R}^{\{0,1\}^m} $ be a unit vector and let $U_{\theta}$ be as defined above. Let $\rho_{\theta}(x) = (\langle e_x, U_{\theta}^{\otimes m} \phi \rangle)^2$ be a probability measure on $\{0,1\}^m$. Let $X = (X_1,\dots,X_m)$ such that $X \sim \rho_{\theta}$. Let $S \subseteq [m]$, $|S|=r$ and $X_S= \prod_{i\in S} X_i$ be a monomial. Then, $\E(X_S)$ is a trigonometric polynomial with frequency at most $r$.
\end{lemma*}
\begin{proof}
    \[ \E\left[\prod_{j \in S} X_j\right] =  \sum_{x \in \{0,1\}^m}  \left(\prod_{j \in S} x_j \right) \rho_{\theta}(x).  \]
    Substituting $\rho_{\theta}(x) = ((U_{\theta}^{\otimes m} \phi)_{x})^2 $, we get 
    \[\E\left[\prod_{j \in S} X_j\right] =  \sum_{x \in \{0,1\}^m}  \left(\prod_{j \in S} x_j \right) ((U_{\theta}^{\otimes m} \phi)_{x})^2\]

    Let $B = \begin{pmatrix}
    0 & 0 \\
    0 & 1
\end{pmatrix}$.

We can then write $x_j = e_{x_j}^T B e_{x_j}$. Furthermore, for an m-bit string $x \in \{0,1\}^m$ and a set $S$, \[\prod_{j \in S} x_j=  e_{x}^T \bigotimes_{j=1}^m C_j e_{x}\] where $C_j = \mathbf{1}_{\{j\in S\}}(B-I) + I$.

Also, let $U_{\theta}^{\otimes m} \phi = u$. Then, $\E\left[\prod_{j \in S} X_j\right] = \sum_{x \in \{0,1\}^m}  e_{x}^T \bigotimes_{j=1}^mC_j e_{x} (u_{x})^2 = u^T \bigotimes_{j=1}^mC_j u$.

Substituting $u = U_{\theta}^{\otimes m} \phi$, we get \[\E\left[\prod_{j \in S} X_j\right]=\phi^T (U_{\theta}^{\otimes m})^T \bigotimes_{j=1}^mC_j U_{\theta}^{\otimes m} \phi = \phi^T  \left(\bigotimes_{j=1}^m U_{\theta}^T C_j U_{\theta} \right) \phi. \] 

But 
\[
    U^T_{\theta} C_j U_{\theta} =
        \begin{cases}
           U^T_{\theta} B U_{\theta}  , &j \in S , \\
             \; I, \quad &j \notin S.
        \end{cases}
\]

$U^T_{\theta} B U_{\theta} = \frac{I + Z\cos{\theta}+J \sin{\theta}}{2}$ where $Z = \begin{pmatrix}
    -1 & 0 \\
    0 & 1
    \end{pmatrix}
    $ and $J = \begin{pmatrix}
    0 & 1 \\
    1 & 0
    \end{pmatrix}$. Each $j \in S$ can contribute at most $1$ to the frequency of $\E[X_S]$. Since $|S|=r$, the frequency of $\E[X_S]$ is at most $r$.

\end{proof}

\begin{lemma*}[Restatement of Lemma~\ref{lem:threshold_avg_trig}] Let $Q$ be a multilinear polynomial on $n$ blocks of $m$ variables, with degree at most $D$. Also, let $0 \le Q \le 1$ on the entire Boolean cube. 
Let \[ R(\theta_1,\dots,\theta_n) = \E_{X^{(i)} \sim \rho_{\theta_i}} [Q(X^{(1)},\dots,X^{(n)})]
\]
where $X^{(i)}$ are blocks sampled independently from $\rho_{\theta_i}$ described in Lemma~\ref{lem:prob_dist_th}.

Then, $0 \le R \le 1$ on $\mathbb{R}^n$ and $R$ is a trigonometric polynomial of total frequency at most $D$.

Furthermore, for every $z \in \{-1,1\}^n$,
  \[ \SF(R)(z/\kappa) = \E_{X^{(i)} \sim \sigma_{k+(1+z_i)/2}}Q(X^{(1)},\dots,X^{(n)}). \]
\end{lemma*}
\begin{proof}
    Since $0 \le Q \le 1$ and $R$ is an expectation of $Q$, $0 \le R \le 1$.

    Let $M$ be a monomial of $Q$. We can factor this monomial $M$ into $n$ monomials $M_i$ such that $M_i$ only contains variables from $i^{th}$ block. Let $S_i \subseteq [m]$ be such that $M_i = \prod_{j \in S_i} x^{(i)}_j$. Hence, $M = \prod_i \prod_{j \in S_i} x^{(i)}_j$.
    Let us first evaluate $\E[\prod_i \prod_{j \in S_i} X^{(i)}_j]$. Since $X^{(i)}$ are sampled independently, $\E[\prod_i \prod_{j \in S_i} X^{(i)}_j] = \prod_i \E[\prod_{j \in S_i} X^{(i)}_j]$.

    Using Lemma~\ref{lem:threshold_trig_expectation}, we get $\E[\prod_{j \in S_i}X^{(i)}_j]$ has trigonometric frequency at most $|S_i|$ and thus contribute at most $|S_i|$ to the frequency of monomial $M$. Let $\sum_{\alpha} c_{\alpha} e^{i\langle \alpha,\theta \rangle}$ be the trigonometric polynomial for monomial $M$ where $\alpha_i$ comes from individual monomial $M_i$. Then, since $|\alpha_i| \le |S_i|$ and since $\sum_i |S_i| \le D$, the total frequency for every monomial is at most $D$. Therefore, total frequency of $R$ is at most $D$.

Since $R$ is trigonometric, $R$ is analytic in a neighborhood of 0. Using Lemma~\ref{lem:prob_dist_th}, we have $\nu_{(1+z)/2} = \rho_0 + \frac{z}{\kappa}(\rho_0')$. Finally using Lemma~\ref{lem:squarefree_approximation}, we get
\[ \SF(R)(z/\kappa) = \E_{X^{(i)} \sim \sigma_{k+(1+z_i)/2}}Q(X^{(1)},\dots,X^{(n)}). \]
\end{proof}

\end{document}